\documentclass[journal]{IEEEtran}

\usepackage{graphicx}
\usepackage{amssymb}
\usepackage{epstopdf}
\usepackage{lscape, amsmath, amssymb, amscd, eufrak, mathrsfs}
\usepackage[subrefformat=parens,labelformat=parens]{subfig}
\usepackage{algorithm}
\usepackage{algorithmic}
\usepackage{multirow}
\usepackage{tikz, pgfplots}
\pgfplotsset{compat=1.3}
\usetikzlibrary{patterns}

\newtheorem{mydef}{Definition}
\newtheorem{theorem}{Theorem}
\newtheorem{lemma}{Lemma}

\newcommand{\myalpha}{\tilde \alpha}
\newcommand{\mysumrate}{generalized normalized sum-rate}

\begin{document}

\title{Leveraging Physical Layer Capabilites: Distributed Scheduling in Interference Networks with \\Local Views}

\author{Pedro~E.~Santacruz,~\IEEEmembership{Member,~IEEE,}
        Vaneet~Aggarwal,~\IEEEmembership{Member,~IEEE,}
        and~Ashutosh~Sabharwal,~\IEEEmembership{Senior~Member,~IEEE}
        }

\maketitle

\begin{abstract}
\boldmath
In most wireless networks, nodes have only limited local information about the state of the network, which includes connectivity and channel state information. With limited local information about the network, each node's knowledge is mismatched; therefore, they must make distributed decisions. In this paper, we pose the following question - if every node has network state information only about a small neighborhood, how and when should nodes choose to transmit? While link scheduling answers the above question for point-to-point physical layers which are designed for an interference-avoidance paradigm, we look for answers in cases when interference can be embraced by advanced PHY layer design, as suggested by results in network information theory.

To make progress on this challenging problem, we propose a constructive distributed algorithm that achieves rates higher than link scheduling based on interference avoidance, especially if each node knows more than one hop of network state information. We compare our new aggressive algorithm to a conservative algorithm we have presented in \cite{santacruz:2013}.  Both algorithms schedule sub-networks such that each sub-network can employ advanced interference-embracing coding schemes to achieve higher rates. Our innovation is in the identification, selection and scheduling of sub-networks, especially when sub-networks are larger than a single link.
\end{abstract}


%
\IEEEpeerreviewmaketitle

\section{Introduction}
\label{Intro}
\IEEEPARstart{T}{he} shared nature of wireless communication networks results in the fundamental problem of dealing with interference from other simultaneous transmissions by co-located flows. The most commonly used technique of managing interference is to avoid it by scheduling transmissions such that the co-located flows do not transmit simultaneously.  Link scheduling inherently assumes that the underlying physical layer architecture is designed to decode a single packet. Link scheduling, both centralized and distributed, has a rich history and continues to be an active area of research, see \cite{tassiulas:1992, gupta:2000, lin:2005, sharma:2006, wu:2007, chaporkar:2008} and references therein.  We pose and study the scheduling problem for the case when the physical layer architecture can embrace interference by using advanced coding methods.

While interference-avoidance  continues to be the near de-facto strategy in wireless networks, it has been known for some time that avoiding interference is not a capacity maximizing strategy for many networks. For example, techniques like multi-user detection \cite{verdu:1998}, Han-Kobayashi coding for 2-user interference channel \cite{han:1981}, and interference alignment for general interference networks \cite{cadambe:2008} are known to yield higher capacity by embracing, not avoiding, interference; e.g., see the book-length exposition~\cite{elgamal:2011} for details. These new ideas have also inspired new standardization activity like Coordinated Multipoint (CoMP)~\cite{irmer:2011}, which uses network MIMO to improve capacity at the edge of the cells. However, almost all such advanced techniques assume extensive knowledge about the network topology, channel statistics, and in many cases, instantaneous channel information, to achieve capacity gains from embracing interference~\cite{elgamal:2011}.  A direct impact of requiring such extensive knowledge at each node is that the resulting network architecture poses scalability limitations -- as the network size grows, the amount of network state information needed at every node also grows proportionally with the number of users in the network.

In this paper, we pose the following problem. If each node in the network has limited information about the network state (connectivity and channel states), i.e., it only knows the network state information about channels and links $h \geq 1$ hops away from it, then what is the capacity maximizing transmission scheme. Note that limited local information problems have been extensively studied in distributed scheduling~\cite{lin:2005, sharma:2006, wu:2007, chaporkar:2008}. However, as mentioned above, all of them assume interference avoidance as their underlying architecture. In our model, the physical layer architecture is not restricted a priori and is allowed to be any feasible scheme, including those which embrace interference~\cite{elgamal:2011}. However, unlike network information theory formulations \cite{han:1981, cadambe:2008, sato:1981}, we are explicitly studying only scalable architectures by limiting network state information at each node. 

The new problem turns out to be extremely hard, and is the generalized version of the distributed capacity problem studied recently in~\cite{aggarwal,aggarwal:2012}.  The formulation in~\cite{aggarwal,aggarwal:2012} shared full network connectivity information with all nodes but assumed only $h$-hop information about the channel state, where $h$ is less than the network diameter. The key (and surprising) result in~\cite{aggarwal} was that a generalized form of scheduling is information-theoretically optimal for many networks. The general scheduling, labeled Maximum Independent Graph Scheduling (MIG Scheduling), schedules \emph{connected sub-networks} larger than a single link, particularly when $h >1$, i.e., nodes know more than one hop of channel information. The size of the sub-networks is such that within each sub-network, an advanced coding scheme can be used since all nodes have enough channel information to operate optimally. Since MIGS assumed that the nodes have full connectivity information, nodes across different sub-networks can coordinate their time of transmission. Thus MIGS is akin to centralized scheduling, but of connected sub-networks potentially bigger than a single link. 

In our work, no global connectivity information is available at any node and thus all decisions have to be truly distributed. To make progress on the challenging new network capacity problem, we use MIGS as our starting point and focus on how sub-networks can be \emph{identified}, \emph{selected} and \emph{scheduled} in a distributed manner.  We began exploring this problem in \cite{santacruz:2013} where we presented a distributed sub-network scheduling algorithm that included a selection step that conservatively ensured that performance gains could be guaranteed.  In this work, we continue our contributions by presenting a new distributed sub-network scheduling algorithm which is shown in simulation to achieve better network performance compared to interference-avoidance, especially when more than one-hop of network state information is made available at each node.

The first step is the identification of sub-networks which can operate optimally with limited local information. That is, if the rest of the network was switched off beyond a small sub-network, there exists enough information to use an advanced interference-embracing coding scheme which will achieve the maximum possible capacity in that sub-network. 
We identify a set of sub-graphs which we denote as $\rho$-cliques (defined later), where $\rho$ depends on the amount of available network information.  By limiting our attention to finding only $\rho$-cliques, we ensure that each sub-network by itself can operate optimally with limited local network view at each node and still guarantee the necessary condition of sub-networks to be used by MIGS.


In the second step, we select a subset of identified $\rho$-cliques in order to maximize network sum-rate. The challenge, like in any distributed problem, is for nodes to reach consensus \emph{locally} such that the global rate is optimized. Towards that end, we propose a new selection algorithm labeled \emph{aggressive} and compare it with our previous algorithm in \cite{santacruz:2013}, which we label \emph{conservative}. Both selection algorithms prune the identified $\rho$-cliques in order to decrease the maximum degree of the graphs made from the identified $\rho$-cliques, which is directly related to increasing the network sum-rate. Thus, the two algorithms use only local graph properties in their decision making. The conservative algorithm is guaranteed to produce schedules which will never achieve rates below interference-avoiding link scheduling, but ends up making conservative decisions for many networks. The aggressive algorithm does not provide provable guarantees but is shown to achieve better sum-rates for several network classes. The last step of scheduling sub-networks is performed using Kuhn's local multicoloring algorithm~\cite{kuhn:2009}, applied to the more general graph structure induced by sub-networks. 

To summarize, the contributions of our work are threefold.  The first contribution is associated with our problem formulation.  While previous work has looked into network performance using normalized sum-rate \cite{aggarwal, aggarwal:2012, kanes}, the formulation in that literature has assumed that only channel state information is locally available and connectivity information is available globally to all users in the network.  In this work we remove the assumption of globally available connectivity information and formulate the problem to characterize a more general form of normalized sum-rate with $\eta$ hops of channel information and $\tau$ hops of connectivity information.  

The second contribution is the design of a new constructive distributed sub-network scheduling algorithms to improve normalized sum-rate performance using local network views.  The new algorithm is labeled \emph{aggressive} and addresses the limitations of the \emph{conservative} algorithms proposed in \cite{santacruz:2013}.  The proposed one-shot algorithm is based on simple heuristics and uses $\eta$ hops of channel state information and $\tau$ hops of connectivity information.  

Finally, our third contribution is the performance analysis of the conservative and aggressive algorithms.  We present a mathematical characterization of the algorithms' performance in terms of normalized sum-rate.  In \cite{santacruz:2013}, we showed that the performance of the conservative sub-network scheduling algorithm is guaranteed to be a non-decreasing function of the number of hops of information available to each node.  In this paper we show through simulations that the aggressive algorithm achieves significant normalized sum-rate performance gains in several important network classes.  Also, to provide a more practical measure of net throughput performance, we compare performance bounds of our algorithms against distributed coloring and maximal scheduling algorithms and show that our proposed aggressive algorithm can provide significant gains in terms of net sum-rates.

The rest of the paper is organized as follows:  in Section~\ref{sec:Related} we discuss the related work, in Section \ref{sec:Model} we present the system model and formalize the problem formulation, and in Section \ref{sec:Subnetwork} we introduce an overview of the distributed sub-network scheduling algorithm.  Sections \ref{sec:Step1}, \ref{sec:Step2}, and \ref{sec:Step3} describe Step~1, Step~2, and Step~3 of the proposed distributed algorithm, respectively.  In Section \ref{sec:Results} we present the results and analysis comparison of both the conservative and aggressive algorithms we have developed and we conclude our paper in Section \ref{sec:Conclusion}.

\section{Related Work}
\label{sec:Related}

Interference management in wireless networks has been widely studied at both the network and physical layers.  From the networking point of view, prior research has focused on developing schemes and algorithms that reduce the computational complexity and allow optimal throughput strategies to be executed in a distributed fashion.  At the physical layer, the focus has been set on finding schemes that aim to approach network capacity.  Approaches at both the networking and physical layer have acknowledged and tried to address the issue of local information in some form or another.  In this section, we review some of the works related to characterizing the effects of local information in wireless networks.  

\subsection{Interference Avoidance: Link Scheduling}
In the seminal work by Tassiulas and Ephremides~\cite{tassiulas:1992}, the capacity of a constrained queueing system for an interference-avoiding physical layer was derived and characterized.  The problem was solved by finding a maximum-weight independent set of nodes in the network graph in each transmission time-slot.  While this scheme yields optimal throughput under the interference-avoidance paradigm, it requires complete information about the connectivity and queue states of the network by all members of the network or, alternatively, a centralized entity that computes the optimal schedule and communicates it to the members of the network in each time-slot.  Furthermore, even when complete information is available, the optimization problem that needs to be solved has high complexity.  The need for complete information and the high computation complexity of the optimal solution make maximum weight-scheduling infeasible for most practical purposes, especially in wireless scenarios with time-varying topologies.  There has been a significant amount of work on reducing the amount of network information required at each node, developing algorithms that are implementable in a distributed manner, and reducing the computational complexity of the algorithms.  Some prior work aims to maximize performance in terms of utility \cite{jiang:2010IT, lee:2009, yi:2008} while others are concerned with improving throughput.  

One set of link scheduling algorithms consists of algorithms that aim to produce \emph{Maximal Matchings} at each time-slot \cite{chaporkar:2008}, \cite{lin:2005}, \cite{sharma:2006}, \cite{wu:2007}.    Also, there are the so-called \emph{Pick-and-compare} policies where, at each time-slot, the current schedule serves as a building block for the next schedule, \cite{sanghavi:2007} and \cite{modiano:2006}.  Both of these sets of algorithms require global knowledge about the connectivity of the network, or at least some predetermined global ordering that allows nodes to avoid conflicts during algorithm execution. In addition, the long convergence time to the optimal schedule can result in increased delay.  


Another class of algorithms under the paradigm of interference-avoidance are the policies that require a \emph{constant-time overhead}, \cite{lin:2006}, \cite{gupta:2009}, and \cite{joo:2009}.  These policies are all based on separating each time frame into a scheduling time-slot and data transmission time-slot.  The design of these algorithms presents an explicit tradeoff between performance and overhead.  Constant-time overhead algorithms also include those that employ Carrier Sense Multiple Access/Collision Avoidance (CSMA/CA), \cite{marbach:2008}, \cite{jiang:2010}, and \cite{ni:2012}.  Finally, more recent work, \cite{joo:2012} and \cite{leconte:2011}, has developed algorithms that introduce the idea of network locality in the process of scheduling.  These algorithms provide \emph{local greedy scheduling} schemes that approximate Greedy Maximal Scheduling \cite{lin:2005} with nodes in the network using only information about themselves and their neighbors.

All algorithms described in this section, and other approaches not categorized here \cite{baker:1982, shah:2002, eryilmaz:2005, sharma:2006Allerton, sarkar:2006, zussman:2008, ying:2009, wu:2010} have the predetermined assumption that receivers in the network decode an incoming transmission successfully only if none of the other links within reception range transmit concurrently.  Some of the implementation issues of the optimal solution presented in \cite{tassiulas:1992} are addressed by the works described above, yet the underlying interference-avoidance physical-layer architecture remains.  The algorithms presented in the works described provide more practical and more easily implementable schemes that can guarantee a fraction of the performance of the scheduling capacity region described in \cite{tassiulas:1992}, but leave the possibility of leveraging advanced physical layers with local information open for exploration.

\subsection{Beyond Interference Avoidance}
Interference-avoidance strategies are often not the optimal approach from a sum-capacity perspective and developments in network information theory describe advanced coding techniques that achieve higher sum-capacity \cite{elgamal:2011}.  In fact, the capacity region of a network with interference-avoidance as the assumed physical layer must lie completely inside the capacity region where all physical layer techniques are available.  Naturally, the major drawback of these results is the need for complete (or almost complete) network information, including connectivity and channel information.  

Several works have addressed this prohibitive requirement by analyzing the performance of networks with only limited information.  In \cite{jafar:2013}, the author considers an interference channel where transmitters do not have channel information and only connectivity information is available.  Network performance under different example topologies is analyzed in terms of degrees of freedom (DoF).  Techniques and results used in the problem of wired networks with linear network coding are applied to the wireless problem.  The work in \cite{naderia:2013} produces results in a similar scenario that assumes connectivity information but no channel information is available to the transmitters.  A study of the capacity region for the 2-user interference channel when each transmitter knows only a subset of the channel gains in the network is presented in \cite{kao:2011}.  These results along with the work in \cite{aggarwal} motivate the argument that it is possible to use advanced coding techniques with limited local knowledge.  We will use the work in \cite{aggarwal} as the launchpad for our work and expand it to capture a more detailed analysis of the impact of local information on network performance.

\section{System Model and Problem Formulation}
\label{sec:Model}

\subsection{Network Model}
In this work, we consider a wireless interference network consisting of $N$ source-destination pairs.  Each source-destination pair is considered a user in the network.  The source nodes, labeled $S_i$, are connected to a subset of the destination nodes, labeled $D_j$, ($i,j \in \{1, ..., N\}$) if the received power at destination $D_j$ from $S_i$ is above some threshold. The set of transmitters that are connected to receiver $D_j$ is labeled $\mathcal I_j$ and receiver $D_j$ is only interested in messages from transmitter $S_j$.  All transmissions from $S_i$, where $i \neq j$, is considered interference at $D_j$.  We assume there is always a connection between $S_i$ and $D_j$, for all $i=j$.  The channel gain between $S_i$ and $D_j$ is denoted $H_{ji}$. The received signal at receiver $D_j$ is 

\begin{equation}
Y_j =\sum_{i: S_i \in \mathcal I_j} H_{ji} X_i + W,
\end{equation}
where $X_i$ is the transmit signal from $S_i$ subject to its average power constraint $P_i$ and $W$ is complex Gaussian noise $\mathcal{CN}(0,1)$.

Associated with the interference network, there is a \emph{conflict graph} $G(V,E)$.  In this conflict graph, a vertex $v \in V$ represents a user in the interference network and an edge $e \in E$ represents interference between those two users, that is, $\{i,j\} \in E$ if $S_i$ is connected to $D_j$ or if $S_j$ is connected to $D_i$.  Figure~\ref{fig:intnetwork} depicts an example interference network and its corresponding conflict graph is illustrated in Figure~\ref{fig:conflictgraph}.  It is important to note that, since our conflict graph is undirected, there are several interference networks that result in the same conflict graph.  While all our results hold for arbitrary conflict graphs, we will use the conflict graph in Figure \ref{fig:conflictgraph} as an illustrative example in the rest of the paper.  Hence, we label the $N$-node graph (in Figure~\ref{fig:conflictgraph}) with a line of size $N-3$ and an attached clique of size 3 as the \emph{line-clique} graph.

\begin{figure}[htb]
\centering
\includegraphics[width=.13\textwidth]{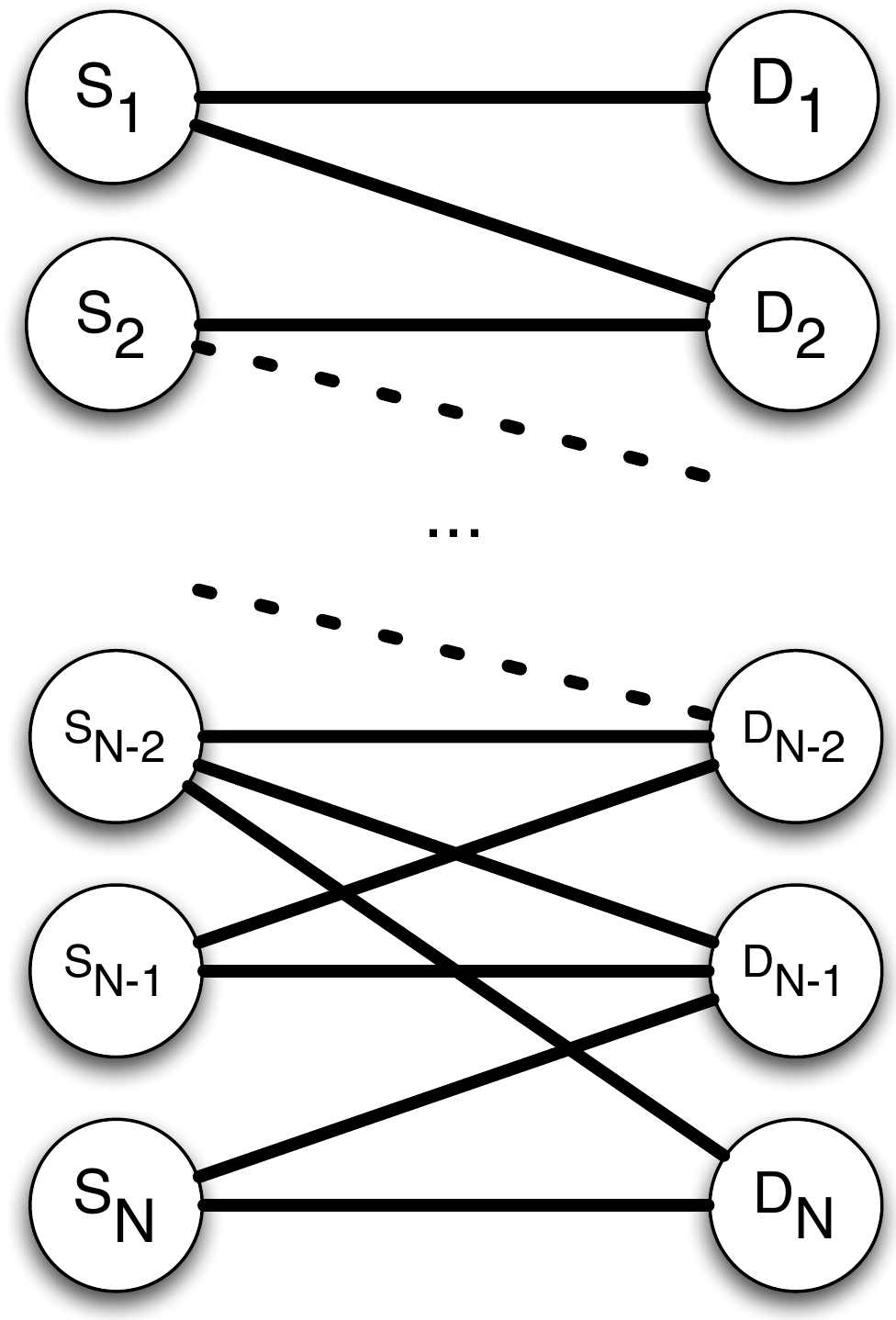}
\caption{Interference Network}.
\label{fig:intnetwork}
\end{figure}

\begin{figure}[htb]
\centering
\includegraphics[width=.3\textwidth]{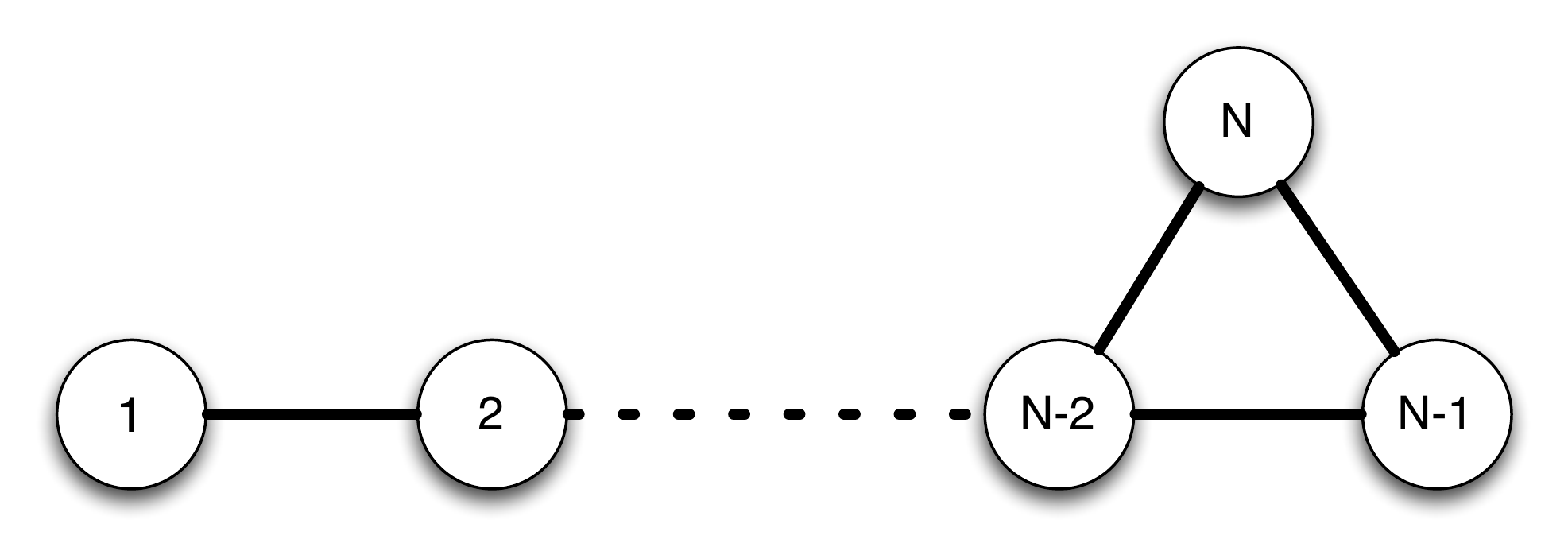}
\caption{Conflict Graph, $G$, associated with Interference Network in Fig. \ref{fig:intnetwork}.  We label it as the \emph{line-clique} graph.}
\label{fig:conflictgraph}
\end{figure}

We note some of the differences between our model and the network model commonly known as the $k$-hop interference model in the network scheduling literature (\cite{lin:2005,lin:2006,sharma:2006,wu:2007} and others).  In the $k$-hop interference model, each node in the interference network can be both a transmitter and a receiver while in our model each node has an assigned role.  Since in our model interference only occurs if a transmitter directly interferes with another receiver, our model is more closely related to the 1-hop, or node-exclusive, interference model.  All traffic in our model is assumed to be single-hop traffic and there exist links in the interference network that are not data flows (users) in the network, i.e., connections between $S_i$ and $D_j$, for $i\neq j$, do not represent data flows, yet a connection exists and interference behavior is inherited from the interference network to the conflict graph.

\subsection{Local View} In~\cite{aggarwal}, local knowledge at any node was modeled as $h$-hops of channel information with that node as the center. However, all nodes were assumed to have full connectivity information. We will study a more general model of network information, giving each node only limited network connectivity and channel state knowledge, as described below. For convenience, we describe the quantified amount of knowledge in the context of the conflict graph representation of the network.

A user is said to have $\tau$ hops of connectivity information if it knows all vertices and edges $\tau$ hops away from it in the conflict graph $G$.  Similarly, a user has $\eta$ hops of channel information if it has knowledge of all channel gains in the interference network for all users $\eta$ hops away in the conflict graph.  Notice that $\eta$ hops of channel knowledge in conflict graph equals  $h = 2\eta+1$ hops of channel knowledge in the interference network.  The same holds true for connectivity information.  Also, we assume that $\tau \geq \eta$ since, in general, channel information is more difficult to obtain than connectivity information.

\subsection{Normalized Sum-rate} Our metric of network performance is a slight modification of the \emph{normalized sum-rate}, $\alpha$, introduced in \cite{aggarwal}, which is the information-theoretic sum-rate achieved normalized by the sum-capacity with full network state information.  
More precisely, a normalized sum-rate of $\alpha(\eta)$ with $\eta$ hops of channel state information in the conflict graph is said to be achievable if there exists a strategy that allows transmission at rates $R_i$ for each flow $i \in \{1, 2, ..., N\}$ with error probabilities going to zero, and satisfying

\begin{equation}
\sum_{i = 1}^{N} R_i \geq \alpha(\eta) C_{sum} - \epsilon
\end{equation}
for all topologies consistent with the local view information, regardless of the realization of the channel gains.  Here $C_{sum}$ is the sum capacity of the network with full information and $\epsilon$ is some non-negative constant independent of channel gains.

In this work, we extend the concept of normalized sum-rate by removing the assumption that full connectivity information is available at every node.  Instead, we quantify both the hops of channel \emph{and} connectivity information available at each node by $\eta$ and $\tau$, respectively.  We propose as our metric of performance the more general form of normalized sum-rate, denoted $\myalpha(\eta, \tau)$, as a function of $\eta$ and $\tau$, such that the following is satisfied 

\begin{equation}
\sum_{i = 1}^{N} R_i \geq \myalpha(\eta, \tau) C_{sum} - \epsilon.
\end{equation}
The problem we address in this work is to characterize the achievable performance in a network with $\eta$ hops of available channel information and $\tau$ hops of available connectivity information.  We tackle this problem by creating schemes that use only $(\eta, \tau)$ hops of information and analyzing their performance.

\section{Sub-network Scheduling}
\label{sec:Subnetwork}

Inspired by the work in \cite{aggarwal}, we generalize the proposed strategy called Maximal Independent Graph (MIG) Scheduling which is information-theoretically optimal for various classes of networks. Information-theoretic optimality means that there exist no other physical layer coding strategy which can achieve higher sum-rates given the amount of knowledge available. In MIG Scheduling, every node in the interference network knows the complete topology of the network and each node is assumed to have $h$ hops of channel information.  The network is separated into sub-networks that can achieve a normalized sum-capacity of 1 with the $h$ hops of channel information (i.e., sub-networks with enough local knowledge to simultaneously transmit in an optimal way).  Thus, the important result from~\cite{aggarwal} is that a generalized form of network scheduling is information-theoretically optimal in many cases. Our approach finds simple distributed algorithms that use only local connectivity information \emph{and} local channel information to find key sub-networks and implement this generalized form of network scheduling.

In MIG Scheduling, the sub-networks that are able to achieve normalized sum-capacity of 1 are labeled \emph{independent sub-graphs}.  MIG scheduling divides the network into $t$ independent sub-graphs, $\mathcal A_1,\ldots, \mathcal A_t$ (not all distinct, for some $t$), and each user $i$ belongs to $d_i$ independent sub-graphs.  An important result in \cite{aggarwal} is that by dividing the network into independent sub-graphs, it is possible to achieves a normalized sum-rate of 

\begin{equation}
\alpha(h) = \min_{i \in 1, 2, ... N}\frac{d_i}{t}.
\end{equation}
The set of independent sub-graphs, $\mathcal A_1,\ldots, \mathcal A_t$, that maximizes the value of $\alpha(h)$ is called the MIG schedule.

The problem of finding the MIG schedule for an arbitrary network is a difficult task, even with complete connectivity knowledge, and particularly challenging with only local connectivity information at each node.  The optimal independent sub-graphs are only known for few topologies and small number of users.  In this work, to answer our posed capacity problem, we focus on the scheduling of independent sub-graphs in a distributed fashion with only $(\tau,\eta)$-hops of knowledge about network state.

We give a general overview of our algorithm to distributedly find independent sub-graphs as required by Independent Graph Scheduling.  In the several next sections we propose an \emph{aggressive} algorithm that is similar to the \emph{conservative} algorithm presented in \cite{santacruz:2013}.  Both algorithms consist of three major steps: 1) identification, 2) selection, and 3) scheduling.  In the first step, we use the available channel and connectivity information to identify all sub-networks of diameter at-most $\rho$ such that each sub-network independently achieves a normalized sum-capacity of 1.  In the second step, we strategically select a subset of these sub-networks.  The selected subset of sub-networks will be the only sub-networks that will be transmitting.  Finally, in the third step we arrange several of these connected sub-networks into independent sub-graphs that still achieve normalized sum-capacity of 1.  The creation of independent graph is done by using a distributed coloring algorithm that assigns a single color to groups of sub-networks.  

Our algorithms are parametrized by $\rho$, which is the maximum diameter of the connected sub-networks being identified.  Given a $\rho$, we assume that each node has at least $\eta = \rho+1$ hops of channel knowledge and either $\tau = 3\rho+3$ or $\tau =3\rho+1$ hops of connectivity information, depending on the algorithm.  For simplicity, we also denote the \mysumrate, $\myalpha(\eta, \tau)$, by a single parameter, $\rho$, and use the symbol $\myalpha(\rho)$.


In Step 1, we leverage the local knowledge available at each node by finding $r$-cliques for $r\le\rho$ which is defined as follows:  

{\mydef An $r$-clique in a graph $G = (V, E)$ is a subgraph, $G[S]$, induced by a subset of nodes $S \subset V$ that satisfies the following three conditions:

\begin{enumerate}
\item Every node in $G[S]$ is at most a distance of $r$ hops away from all other nodes in $G[S]$.
\item The diameter of $G[S]$ is $r$.
\item There is no $S' \subset V$ that also satisfies Conditions 1 and 2 and such that $S \subset S'$.  In other words, $G[S]$ is a maximal subgraph.
\end{enumerate}
}
Note that a single node is a graph by itself and a $0$-clique according to the above definition.

Step 1 consists of identifying $r$-cliques, for $r = 0,...,\rho$, in the conflict graph, $G$.  After the $r$-cliques have been identified, Step 2 consists of selecting a subset of the identified $r$-cliques and consolidating the selected $r$-cliques into single vertices to generate a consolidated graph, $G_\rho$, where each vertex represents an $r$-clique, $r = 0,...,\rho$, from the conflict graph, $G$.  An edge exists between two vertices in the consolidated graph if there exists an edge between members of the two cliques in the original conflict graph.

Step 3 of the general procedure is performed by applying the distributed multicoloring algorithm by Kuhn \cite{kuhn:2009} to the consolidated graph, $G_\rho$, which results in the assignment of time slots to each one of the cliques.  The set of cliques with the same color are defined as an \emph{independent clique set}.  An independent clique set achieves $\alpha(\eta) = 1$ because each clique achieves $\alpha(\eta)=1$ and the cliques do not interfere with each other.  When we assign a time slot to each of the independent clique sets we create a scheme for Independent Graph Scheduling.  We have chosen Kuhn's multicoloring algorithm because it requires only one round of communication and ensures that each node in the graph being colored receives at least a fraction $1/(\Delta +1)$ of the total colors assigned.  We note that our metric of normalized sum-rate is directly related to the time slots assigned to the worst-case user~\cite{aggarwal}. Thus, given a fixed number of cliques containing a specific node, it is desirable to use the multicoloring algorithm in consolidated graphs which have smaller maximum degree, $\Delta$.

Finally, we note that to derive our results we do not need to state the form of optimal coding methods used by each node in the identified sub-networks. The fact that we can analyze sum-rate without explicitly defining coding methods is possible due to our choice of normalized sum-rate as a metric. 

\section{Step 1: Identification ($G \rightarrow G^-_\rho(v)$)}
\label{sec:Step1}

We begin the formal description of our algorithm by describing in detail the procedure followed in the identification step.  Consider a conflict graph $G$ and a parameter $\rho$.  We assume each user in the network has $\eta = \rho +1$ hops of channel information and $\tau = \rho +1$ hops of connectivity information.  In Step~1, for a given $\rho$, each node identifies $r$-cliques, $r = 0,..., \rho$; this identification can be accomplished with $\rho+1$ hops of connectivity information.  We are interested in these $r$-cliques because, with the available channel information, it is ensured that each $r$-clique can achieve $\myalpha(\rho) = 1$.  These potential $r$-cliques are the candidates to ultimately be represented by a vertex in the consolidated graph $G_\rho$.

Since each node has a different local view of the conflict graph, $G$, the potential cliques discovered by each node will be different. Thus, in Step 1, each node will generate a temporary graph where the potential cliques it sees are turned into vertices.  We will denote the temporary graph from the point to view of node $v$ as $G^-_\rho(v) = (W^-(v), F^-(v))$, which is described as follows.  The set of vertices $W^-(v)$ represents all the $r$-cliques ($r\le \rho$) in the part of the graph known to node $v$ with $\rho+1$ hops of connectivity information.  Each node in $w \in W^-(v)$ maps to a set of nodes in the original conflict graph; we denote that set of nodes in the conflict graph represented by vertex $w$ as $nodes(w)$.  An edge exists between two vertices in $G^-_\rho(v)$, $w_1$ and $w_2$, if there is an edge between a member of $nodes(w_1)$ and a member of $nodes(w_2)$ in $G$ or if a member of $nodes(w_1)$ is also a member of $nodes(w_2)$.

Consider an example to illustrate the construction $G \rightarrow G^-_\rho(v)$ with the parameter $\rho =1$ using the example original conflict graph $G$ in Figure \subref*{fig:linecliquecliqueid}.

\begin{figure}[h]
	\begin{minipage}{.49\linewidth}
	\centering

	\subfloat[Line-clique graph with $N$ nodes, $G$ \label{fig:linecliquecliqueid}]{%
      	\includegraphics[width=\textwidth]{figs/examplecg.pdf}
    	}

	\end{minipage}	
	\begin{minipage}{.49\linewidth}
	\centering
    	\subfloat[$G^-_1(1)$ \label{fig:cliqueidn1}]{%
      	\includegraphics[width=\textwidth]{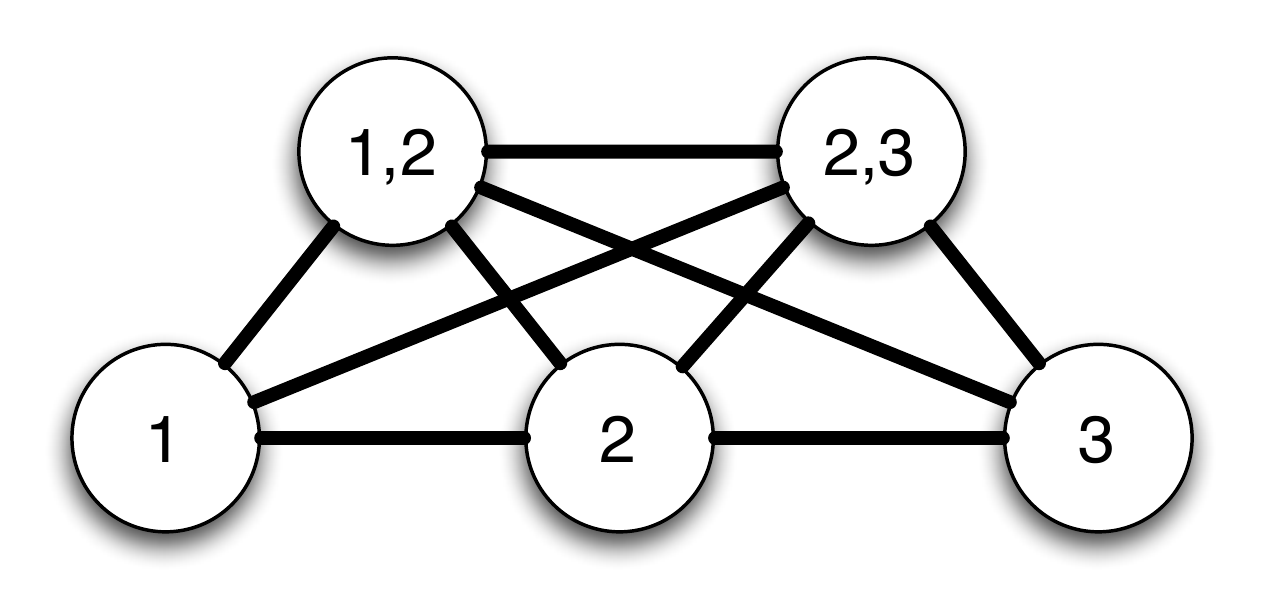}
    	}

	\end{minipage}
	\par\medskip
        \begin{minipage}{.49\linewidth}
	\subfloat[$G^-_1(2)$ \label{fig:cliqueidn2}]{%
      	\includegraphics[width=\textwidth]{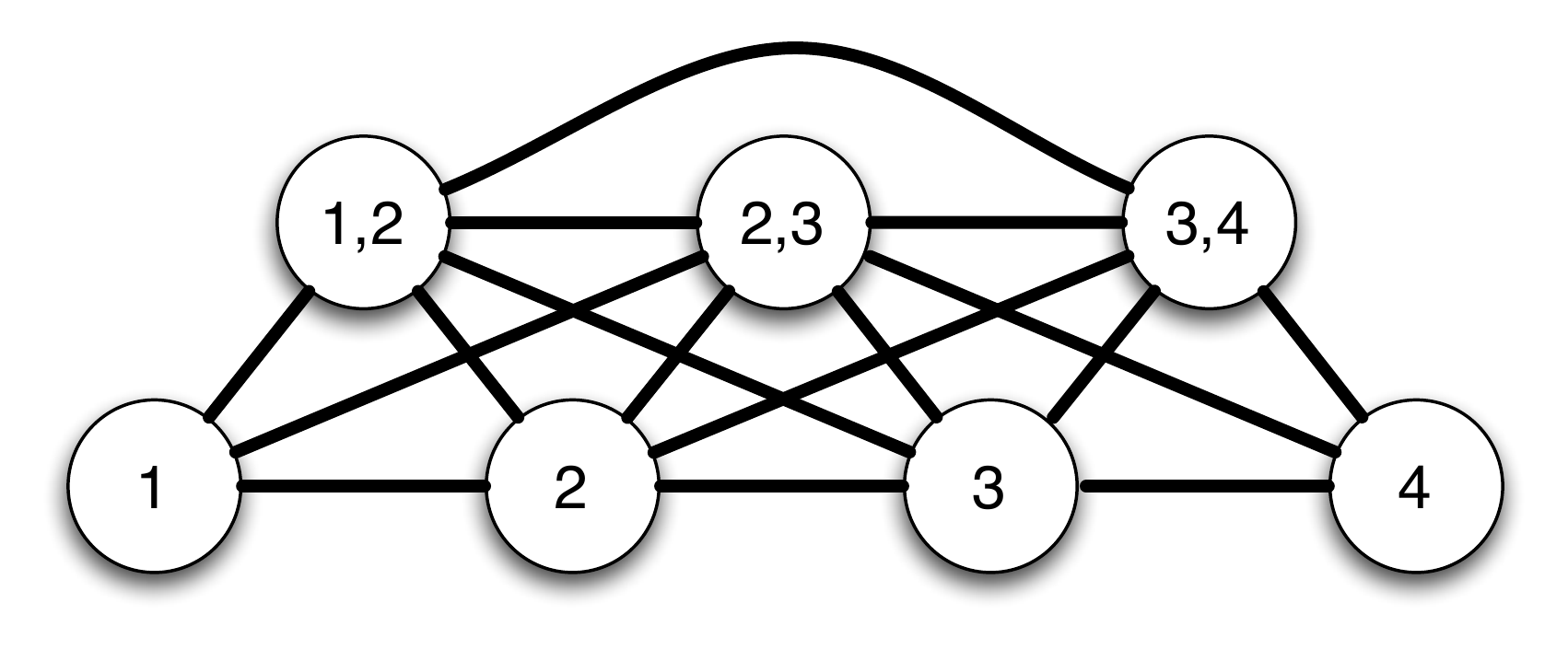}
    	}

	\end{minipage}
	\begin{minipage}{.49\linewidth}
	\subfloat[$G^-_1(N)$ \label{fig:cliqueidnN}]{%
      	\includegraphics[width=\textwidth]{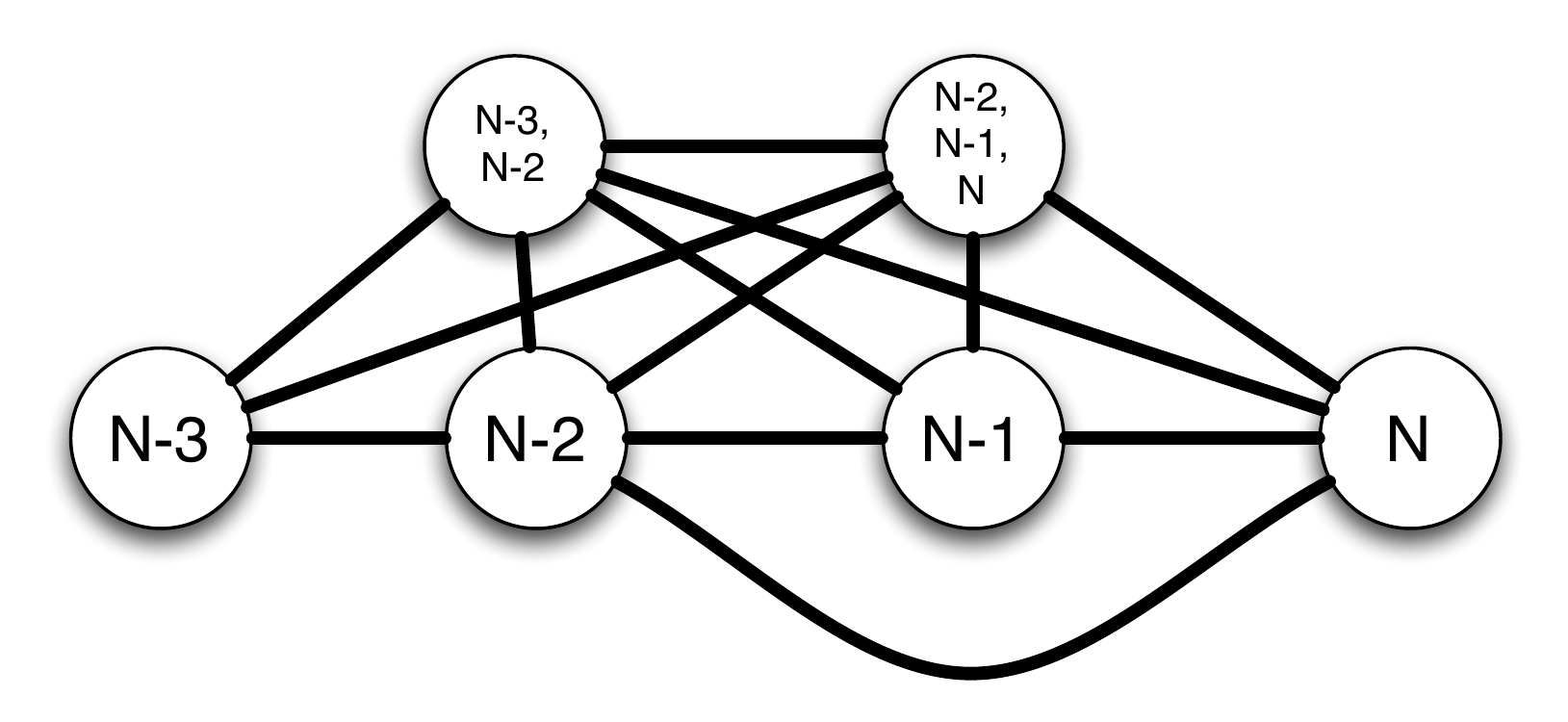}
    	}

	\end{minipage}

\caption{Example of Step 1: Clique Identification for nodes 1, 2, and N with $\rho = 1$}
\label{figs:step1}
\end{figure}

%
%
%


Figure \subref*{fig:cliqueidn1} shows what the graph of all potential vertices looks like from the point of view of Node 1, which has 2 hops of connectivity knowledge.  The vertices are labeled according to their corresponding set of nodes from the original conflict graph (in other words, the label of node $w$ is $nodes(w)$).  As we can see, Node 1 observes 5 potential cliques, three 0-cliques ($\{1\},\{2\}, \{3\}$) and two 1-cliques ($\{1,2\}, \{2,3\}$).  There exists an edge between the vertices labeled $\{1\}$ and $\{1,2\}$ because $\{1\}$ is present in both vertices and because there is an edge between Node 1 and Node 2 in the original conflict graph.  Similarly, there are edges between $\{1, 2\}$ and $\{2\}$ and between $\{1\}$ and $\{2\}$ and so on.

Figure \subref*{fig:cliqueidn2} depicts the graph $G^-_1(2)$.  In this case, Node 2 in $G$ sees four 0-cliques: $\{1\}$, $\{2\}$, $\{3\}$, and $\{4\}$.  Also, Node 2 sees three different 1-cliques $\{1, 2\}$, $\{2, 3\}$ and $\{3, 4\}$ that could be formed so there is a total of 7 vertices in the graph $G^-_1(2)$. The edges are generated following the rules explained earlier.  Similarly, Figure \subref*{fig:cliqueidnN} describes the graph $G^-_1(N)$.  In this last case, notice that the vertex $\{N-2, N-1, N\}$ was created since it forms a 1-clique and node $N$ belongs to it.  An important point is the exclusion of the sets $\{N-1, N\}$ and $\{N-2, N\}$.  These sets are not included as vertices in $G^-_1(N)$ because we have defined an $r$-clique as a maximal subset and both $\{N-1, N\}$ and $\{N-2, N\}$ are 1-cliques superseded by the 1-clique $\{N-2, N-1, N\}$, therefore $\{N-2, N-1, N\}$ is the only set that becomes a vertex in $G^-_1(N)$.


The clique identification process is easily extended for any $\rho >1$ by identifying all $r$-cliques, for $r = 0, ... , \rho$.  For example, if $\rho = 2$, $G^-_2(v)$ would consist of the full $G^-_1(v)$ plus all the $2$-cliques in the 2-neighborhood of $v$, along with their respective edges.  As we have mentioned, a larger $\rho$ would increase the minimum amount of information required at each node.  Also, finding maximal $r$-cliques is in general a hard problem, but since our goal is to leverage local information, we primarily concentrate on the cases of small $r$.

Up to this point, we have identified cliques that are made up of nodes that can transmit simultaneously in an optimal manner with the available local knowledge. However, in order to assign time slots to each one of these cliques, there are two problems that need to be addressed.  First, there is the issue that each node now has a graph with a maximum degree which is significantly higher the maximum degree of the original conflict graph.  Due to the use of Kuhn's algorithm as our scheduling step, the maximum graph degree and normalized sum-rate achieved by our scheme are intimately related, so our goal is a consolidated graph, $G_\rho$, with small degree.  The second issue is the fact that we need to ensure that a distributed coloring algorithm does not lead to coloring conflicts especially since the graphs seen by different nodes differ so much from each other.  In the next section we will select which vertices of the graph $G^-_\rho(v)$ should remain and which should be pruned to reduce the maximum degree of the final consolidated graph, $G_\rho$, and to ensure that there will be no conflict in the use of a distributed coloring algorithm.

\section{Step 2: Selection ($G^-_\rho(v) \rightarrow G_\rho(v)$)}
\label{sec:Step2}

In this section, we will describe the selection step, which consists of selecting which of the potential vertices in $G^-_\rho(v)$ identified by each node in Step 1 will be pruned and which will be kept in their own view of the final consolidated graph $G_\rho(v)$.  In this work, we propose a different approach to Step 2 than the one presented in \cite{santacruz:2013}.  The approach presented previously is a conservative selection algorithm that allows each user to be represented by only one vertex in the consolidated graph and ensures that the normalized sum-rate of the network never decreases.  The alternative approach presented here allows each user to be represented by more then one vertex in the consolidated graph, and while strict guarantees cannot be provided, the gains in normalized sum-rate are significant in some important classes of graphs, especially with small amounts of local knowledge.  We first highlight the key characteristics of the conservative selection algorithm and then proceed to describe the aggressive selection algorithm in detail.  


\subsection{Conservative Selection Algorithm}
\label{sec:Conservative}

In Step 1, we created a graph $G^-_\rho(v)$ that consists of vertices that represent cliques from the original conflict graph and can simultaneously transmit in an optimal way.  As discussed above, graphs $G^-_\rho(v)$ can have the maximum degree which is higher than $G$.  The increase in maximum degree is expected; since cliques of nodes are now transmitting simultaneously, their joint interference footprint is expected to grow.  To guarantee improvement in normalized sum-rates,  the simplest way is a conservative selection algorithm that satisfies two properties:

\begin{enumerate}
\item Each node $v$ from the conflict graph $G$ is represented by \emph{only one} node in the consolidated graph $G_\rho(v)$
\item The degree of the vertex that represent $v$ in the consolidated graph $G_\rho(v)$ is less than or equal to the degree of $v$ in the conflict graph $G$.
\end{enumerate}
Please note that we say user $v$ in $G$ is represented by vertex $w$ in $G_\rho(v)$ if $v \in nodes(w)$.
The two simple properties described above ensure that the procedure will achieve a normalized sum-rate of $\myalpha(\rho) = 1/(\Delta_{G_\rho}+1)$, where $\Delta_{G_\rho}$ is the maximum over all maximum degrees of the $G_\rho(v)$ graphs.  The distributed implementation of the conservative selection requires $\tau = 3\rho+3$ hops of connectivity information and is labeled $\mathcal A_1(3\rho+3)$.  For a detailed description of the conservative selection algorithm we refer the reader to \cite{santacruz:2013}.

\subsection{Aggressive Selection Algorithm}
\label{sec:Aggressive}

Now we present a new approach to selecting which vertices from the graphs $G^-_\rho(v)$ should be carried over to graphs $G_\rho(v)$.  In \cite{santacruz:2013} we proved that the conservative selection algorithm in the previous subsection ensures that the normalized sum-rate achieved by the network is always greater than or equal to the normalized sum-rate using distributed coloring, for all $\rho$ and for any arbitrary graph.  However, since it must provide this strict guarantee, it tends to be overly conservative and loses potential gains in large classes of graph. To address this issue, we propose a second clique selection algorithm that is more aggressive.

The Aggressive Selection Algorithm relaxes the two major constraints of the conservative algorithm: 1) it allows nodes from the original conflict graph to be represented by more than one vertex in the consolidated graph and 2) the degrees of the vertices being kept for the consolidated graphs are allowed to be larger than the degrees in the conflict graph of the nodes that make up the vertices.  We have mentioned that graphs with larger maximum degrees are undesirable, so the aggressive algorithm provides a heuristic that balances the maximum degree of the consolidated graphs with the number of vertices representing each node.  Recall that the normalized sum-rate of a network is the fraction of active time slots of the worst-case node in the network.  Using our proposed distributed procedure, this is simply $\min_{v\in V} a(v)/\Delta_{G_\rho}$, where $a(v)$ is the number of vertices in $G_\rho(v)$ representing node $v$ from the original conflict graph.  As long as the number of vertices representing each node in the network increases enough to counteract for the increase in the maximum degree of the consolidated graph, gains in normalized sum-rate can be achieved.  We now describe the heuristic of the Aggressive Selection Algorithm.

We begin with some assumptions about the amount of network information available to each node in the network.  For purposes of exposition, we will describe the general idea of the aggressive algorithm by assuming complete connectivity information and later show that only $3\rho +1$ hops of connectivity information is needed.  We denote this centralized aggressive selection algorithm by $\mathcal A_2(Full)$ and the distributed form of the aggressive selection is denoted $\mathcal A_2(3\rho+1)$.  Also, we consider the idea of a temporary graph $G^-_\rho$.  In Step 1 we described the process of each node obtaining $G^-_\rho(v)$, the graph $G^-_\rho$ can be described as a ``centralized temporary graph" with full topology information, but still forming cliques of at most diameter $\rho$.  The graph $G^-_\rho$ is a single graph that contains all the possible $r$-cliques, $r = 0, ..., \rho$, in the original conflict graph $G$.

\subsubsection{Example}
We present an example using the line-clique graph depicted in Figure~\ref{fig:conflictgraph} and begin with the case where $\rho=1$.  Consider the graph $G^-_\rho$ and the process the centralized entity uses to decide which vertices to keep and which vertices to prune to make the consolidated graph $G_\rho$.

The basis for the aggressive selection algorithm is quite simple.  We wish to keep as many vertices from $G^-_\rho$ as possible but there are some vertices that create a lot of interference and are somewhat redundant.  For this reason, a vertex, $w$, representing a set of nodes, $nodes(w)$, is removed if every member of $nodes(w)$ appears more twice somewhere else in the vertices of graph $G^-_\rho$.  For example, consider the 0-clique, $\{2\}$, since Node 2 appears in the cliques consisting of $\{1, 2\}$ and $\{2, 3\}$, the 0-clique $\{2\}$ is removed.  The intuition behind the selection process is that we wish to remove nodes to avoid interference and increasing degrees, but at the same time let every user in the original graph have enough contributions in order to have increasing normalized sum-rate.  The aggressive selection step is illustrated in Figure~\subref*{figs:linecliqueremoved}.  With these nodes removed, the final graph that is scheduled is shown in Figure~\subref*{figs:linecliquefinal}.

\begin{figure*}
	\begin{minipage}{.31\linewidth}
	\centering
	\subfloat[$G^-_\rho$ for $\rho = 1$ \label{figs:linecliquecliques}]{%
	\includegraphics[width=2in]{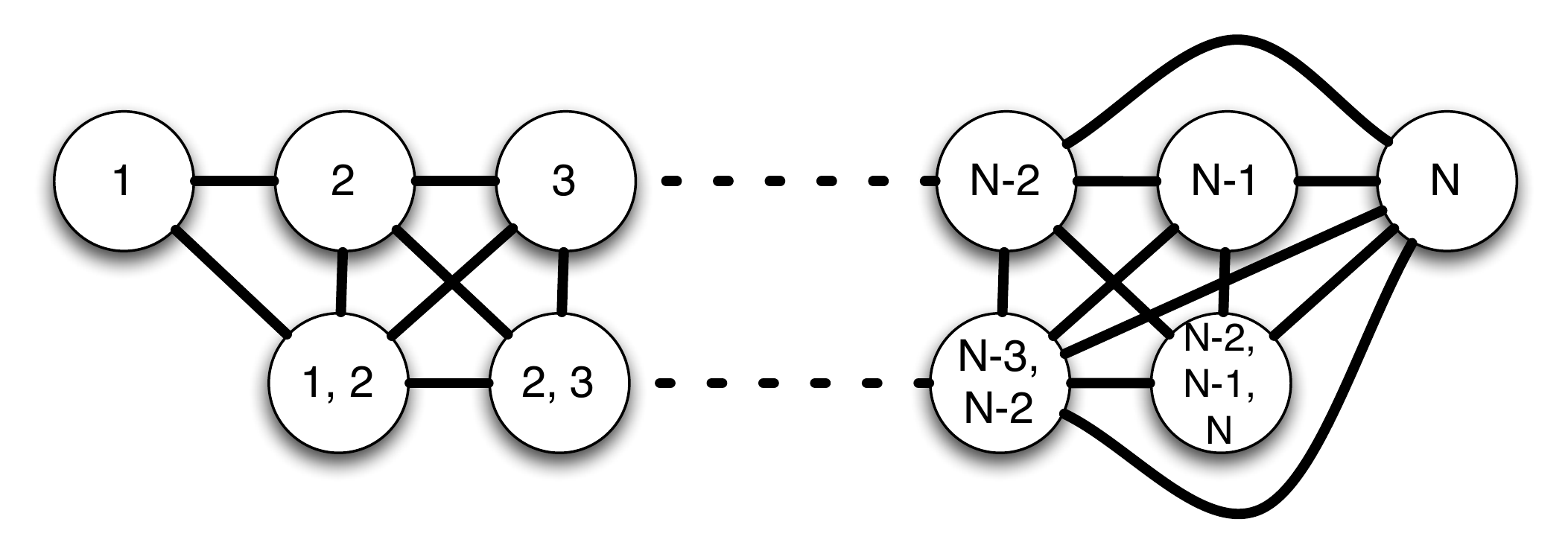}
    	}
	\end{minipage}
	\begin{minipage}{.31\linewidth}
	\centering
    	\subfloat[Cliques that appear twice elsewhere are removed \label{figs:linecliqueremoved}]{%
	\includegraphics[width=2in]{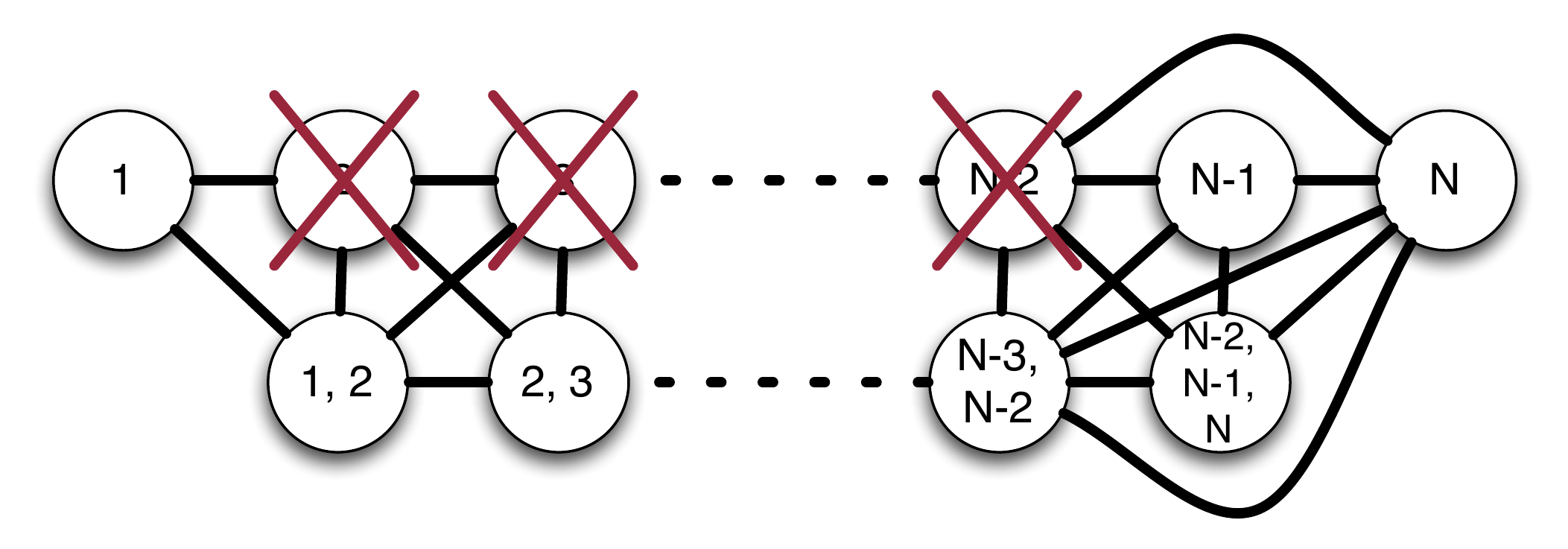}
    	}
	\end{minipage}
	\begin{minipage}{.31\linewidth}
	\centering
	\subfloat[Final graph, $G_1$, after aggressive algorithm \label{figs:linecliquefinal}]{%
	\includegraphics[width=2in]{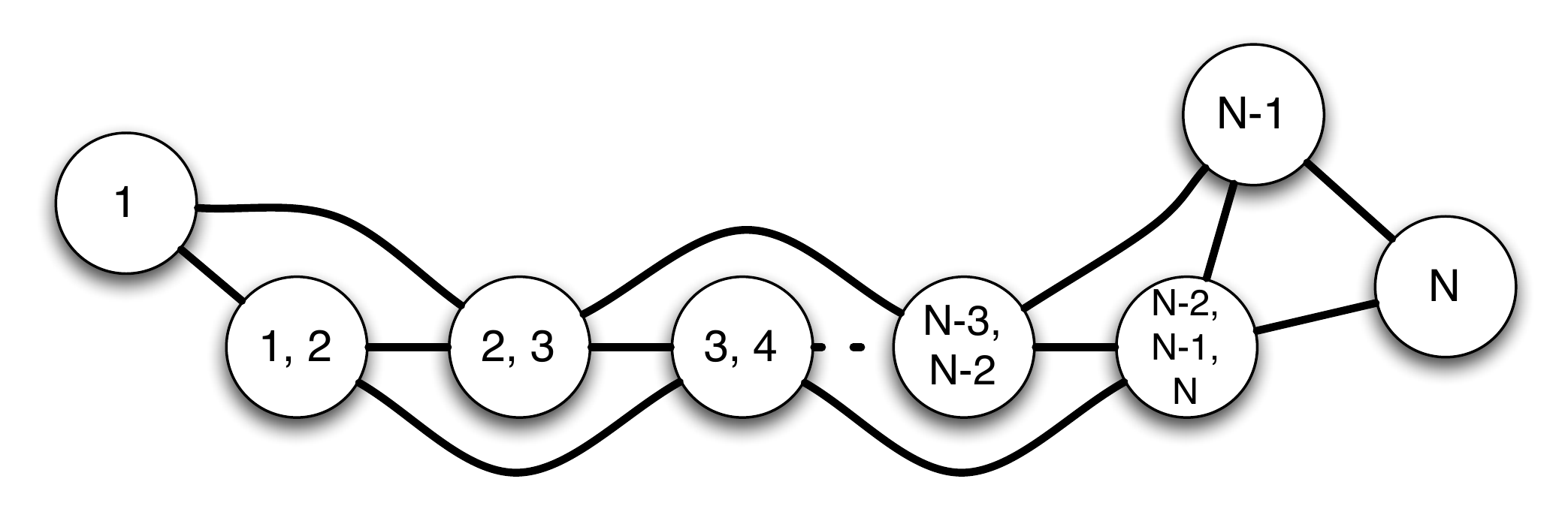}
    	}
	\end{minipage}

\caption{Example of Aggressive Selection Algorithm, $\rho = 1$}
\label{figs:conservative}
\end{figure*}


%
%
%

%
%
%

Now, consider the algorithm when we begin by letting $\rho = 2$.  First, all $1$-cliques and $2$-cliques are generated as in Figure \subref*{figs:linecliquecliquesr2}.  For clarity, we omit the edges is these graphs.  Now, starting with the $0$-cliques (i.e., the single nodes), they are removed if every member is present 2 times in any of the 1-cliques or 2-cliques.  There is an exception to the removal of 0-cliques that establishes that all 0-cliques representing a node of degree 1 remain in the graph even if they are appear twice elsewhere.  This is to avoid pathological cases of possible starvation of end nodes.  In this case, clique $\{1\}$ is not removed.  We define the set of nodes with degree 1 in the conflict graph $G$ as $O_\Lambda$.

Next, the 1-cliques are removed if they appear 2 times in the set of 2-cliques.  It is important to note that $i$-cliques are removed only if they appear twice in the set of $j$-cliques, for all $j>i$.  After performing this operation, illustrated in Figure \subref*{figs:linecliqueremovedr2}, the final graph is shown in Figure \subref*{figs:linecliquefinalr2}.


Notice that the achievable normalized sum-rate in our example line-clique with $\rho = 0$ was $\alpha = 1/4$, with the aggressive selection algorithm when $\rho=1$ then $\alpha_2(\rho) = 2/5$, and when $\rho=2$ then $\alpha_2(\rho) = 3/7$.  Also, compare to the conservative algorithm which achieves $\alpha_1(\rho) = 1/3$, for both $\rho = 1$ and $\rho =2$.  This increase in normalized sum-rate exemplifies the advantages of the aggressive algorithm.


\begin{figure*}
	\begin{minipage}{.31\linewidth}
	\centering
	\subfloat[$G^-_\rho$ for $\rho = 2$ \label{figs:linecliquecliquesr2}]{%
	\includegraphics[width=2in]{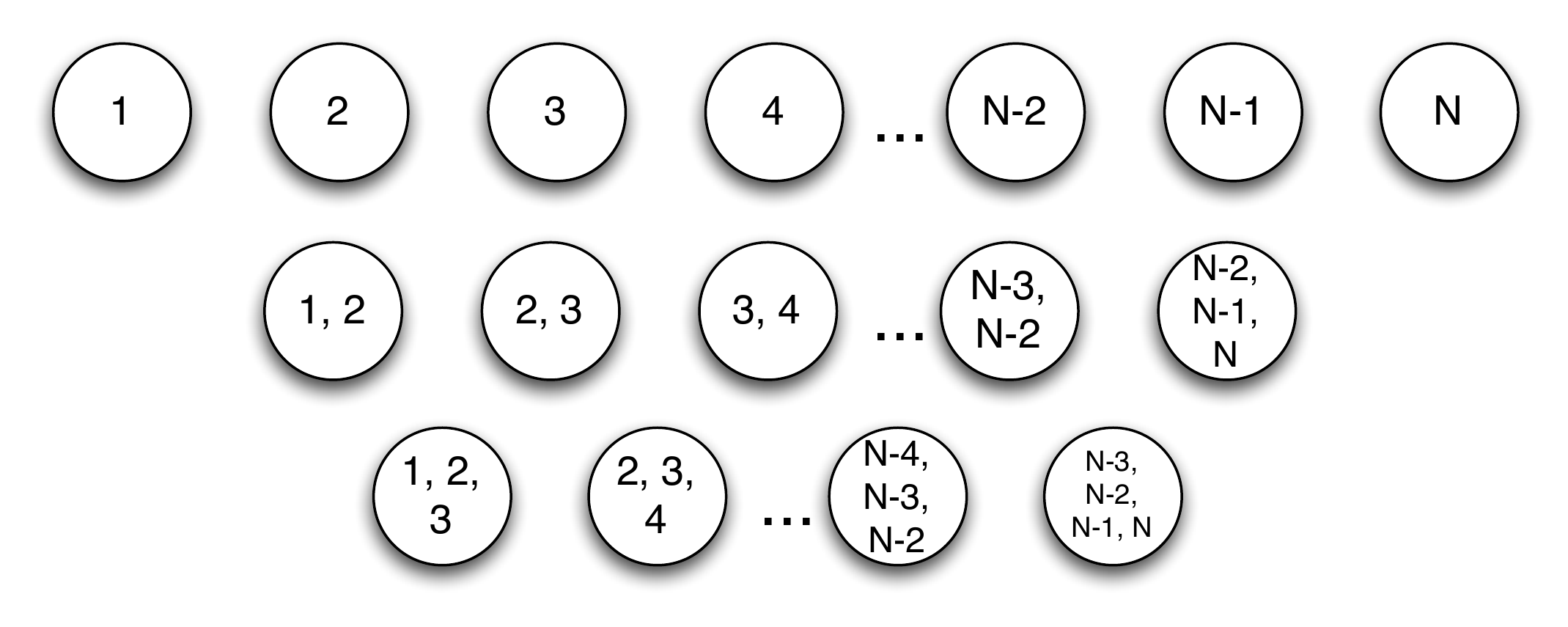}
    	}
	\end{minipage}
	\begin{minipage}{.31\linewidth}
	\centering
    	\subfloat[Cliques that appear twice elsewhere are removed \label{figs:linecliqueremovedr2}]{%
	\includegraphics[width=2in]{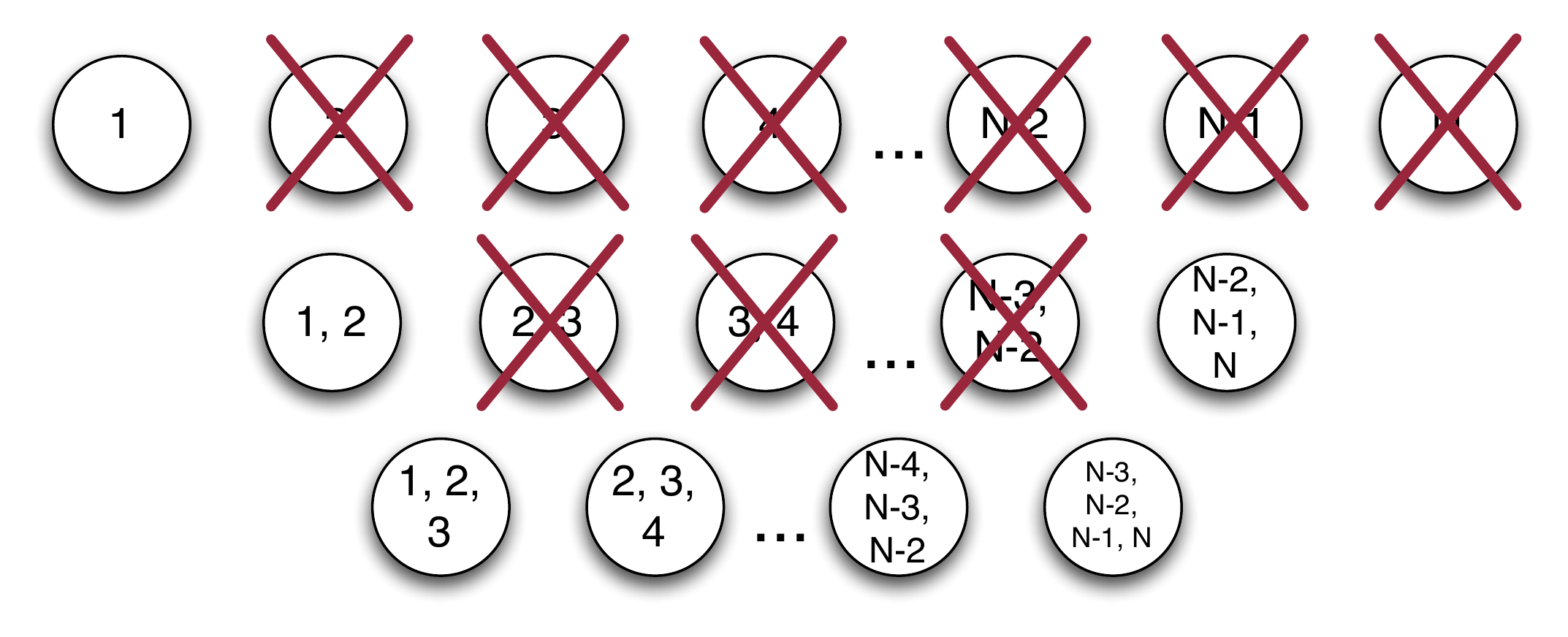}
    	}
	\end{minipage}	
	\begin{minipage}{.31\linewidth}
	\centering
	\subfloat[Final graph, $G_2$, after aggressive algorithm \label{figs:linecliquefinalr2}]{%
	\includegraphics[width=2in]{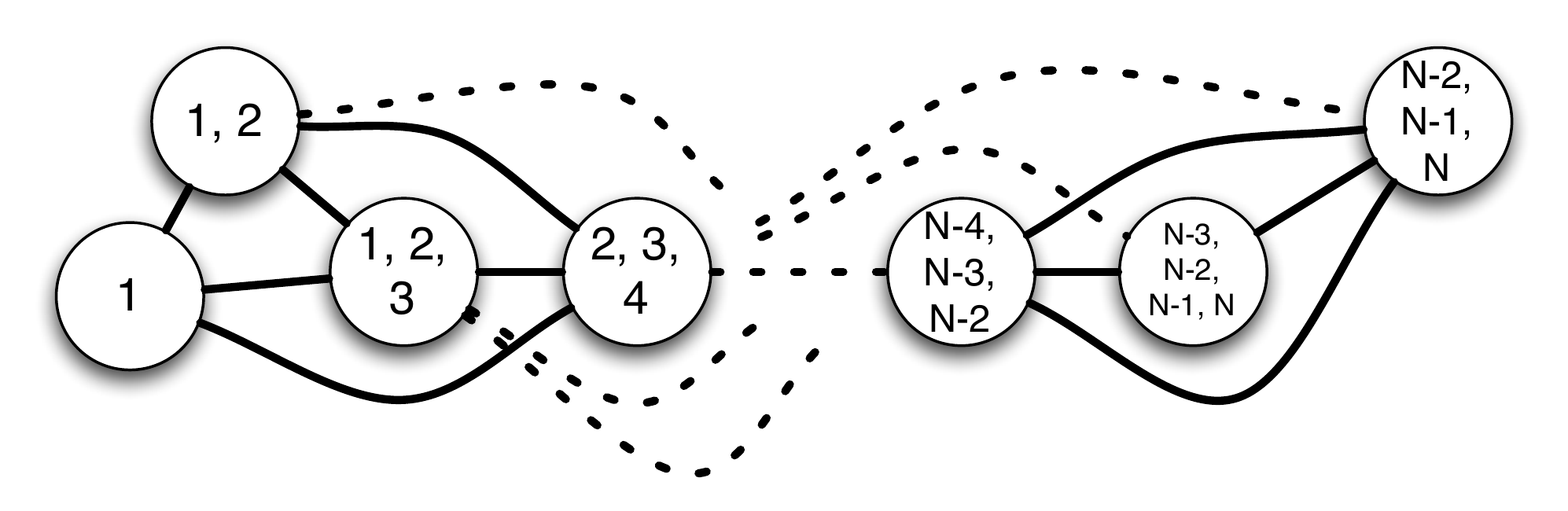}
    	}
	\end{minipage}
	
\caption{Example of Aggressive Selection Algorithm, $\rho = 2$}
\label{figs:aggressiver2}
\end{figure*}


%
%

%
%
%
%
\subsubsection{Distributed Aggressive Algorithm}
We now show that the centralized aggressive selection algorithm described above can be performed in a distributed manner.  Let $G_\rho$ be the consolidated graph of the centralized algorithm and let $G_\rho(v)$ be the consolidated graph from node $v$'s point of view of a distributed algorithm to be described in this subsection.  The requirement that must be fulfilled is that each node performing Kuhn's multicoloring algorithm on $G_\rho(v)$ is equivalent to applying the multicoloring algorithm to $G_\rho$.  Because Kuhn's algorithm only requires knowledge about the 1-hop neighbors, our individual nodes are only interested in knowing 1-hop information about $G_\rho$. In other words, our objective is to show that the neighborhood of vertices containing each node $v$ in $G_\rho(v)$ is identical to their neighborhood in $G_\rho$.  We now describe the distributed aggressive selection algorithm.

We assume that each node has only $3\rho +1$ hops of connectivity information and that each node will be performing independent actions.  Based on the removal heuristic of the centralized algorithm, each node will remove nodes that appear twice somewhere else in the graph they observe using the algorithm described in Algorithm \ref{distributedaggressive}.  We define the set $O_{3\rho+1}(v)$ as the set of nodes of degree 1 in the $3\rho+1$-neighborhood of node $v$ in the original conflict graph $G$.  Also define $A_u^v(s)$ as the number of $s$-cliques representing node $u$ in $G^-_\rho(v)$.  The summary of the Distributed Aggressive Selection Algorithm is shown in Algorithm \ref{distributedaggressive}.

\begin{algorithm}
\caption{Distributed Aggressive Selection $\mathcal A_2(3\rho+1)$}
\label{distributedaggressive}
\textbf{Input:}  Graphs $G^-_\rho(v)$ for each node $v\in V$  with $3\rho+1$ hops of topology information
      	\begin{algorithmic}[1]
		\STATE  The 0-cliques representing members of $O_{3\rho+1}(v)$ are ensured remain in $G_\rho(v)$.
		\FOR{$r =0$ \TO $\rho -1$}
			\STATE Consider vertex $w \in W^-(v)$ (except those representing a node in $O_{3\rho+1}(v)$) that represents an $r$-clique in G.  A vertex $w \in W^-(v)$ is removed from $G^-_\rho(v)$ if for every $u \in nodes(w)$
			\begin{align}
		 	A_u^v(s) \geq 2, \mbox{for }s>i  \mbox{ if } u \notin O_{3\rho+1}(v) \nonumber\\
		 	A_u^v(s) \geq 1, \mbox{for }s>i  \mbox{ if } u \in O_{3\rho+1}(v)\nonumber
			\end{align}
		\ENDFOR
		\STATE The graph $G^-_\rho(v)$ is updated by removing all identified nodes and their connecting edges.  The final result is a graph $G_\rho(v)$.
		\end{algorithmic}
\end{algorithm}


\subsubsection{Consistency of $\mathcal A_2(3\rho+1)$}
We wish to establish that the $\mathcal A_2(3\rho+1)$ algorithm is a valid distributed implementation of algorithm $\mathcal A_2(Full)$.  Since our final objective is to schedule $G_\rho$, all we have to do to establish the validity of $\mathcal A_2(3\rho+1)$ is to show that the neighborhood of cliques containing each node $v$ in $G_\rho(v)$ is identical to their neighborhood in $G_\rho$.  To do this we present the following theorem.

\begin{theorem}
Let each node $v\in V$ have $3\rho+1$ hops of connectivity knowledge.  The neighborhood of every $v\in nodes(w)$ for every $w \in G_\rho(v)$ is identical to the neighborhood of $w$ in $G_\rho$.
\end{theorem}

\begin{proof}
We will break up the proof into two parts.  First, we show that, from the point of view of a node $v$, the neighborhoods of every vertex representing $v$ in $G^-_\rho$ and $G^-_\rho(v)$ are identical, then we show that the cliques removed from each neighborhood are the same.


\begin{lemma}
Let each node $v\in V$ have $3\rho+1$ hops of connectivity knowledge.  The neighborhood of every $w \in G^-_\rho(v)$ such that $v\in nodes(w)$ is identical to the neighborhood of $w$ in $G^-_\rho$.
\end{lemma}
\begin{proof}\emph{(Lemma 1)}
Let $\Gamma_{G^-_\rho}(w)$ be the neighborhood of vertex $w$ in $G^-_\rho$.  Also, consider $G^-_\rho(v)$, where $v$ is a member of $nodes(w)$.  There exists a vertex $w'$ in $G^-_\rho(v)$ such that $nodes(w)=nodes(w')$, since there are at least $\rho$ hops of topology knowledge.  Now, let $\Gamma_{G^-_\rho(v)}(w')$ be the neighborhood of $w'$ in $G^-_\rho(v)$.  Because $v$ has $3\rho+1$ hops of knowledge and the cliques being made are of, at most, diameter $\rho$, $v$ has perfect knowledge of all cliques of up to diameter $\rho$ that include nodes that are at most $2\rho+1$ hops away.  Now, every vertex in $\Gamma_{G^-_\rho}(w)$ represents nodes that are at most $2\rho +1$ hops away from $v$ since each vertex in $\Gamma_{G^-_\rho}(w)$ has at most diameter $\rho$.  Therefore $\Gamma_{G^-_\rho}(w)$ = $\Gamma_{G^-_\rho(v)}(w')$.
\end{proof}


Now that we have shown all the correct neighborhood cliques are generated, we continue by showing that the correct cliques are removed as well.


\begin{lemma}
A vertex in $\Gamma_{G^-_\rho(v)}(w')$ is removed if and only if it is also removed from  $\Gamma_{G^-_\rho}(w)$.
\end{lemma}

\begin{proof}(Lemma 2)
We begin the proof by proving the forward direction, i.e., if a vertex is removed from $\Gamma_{G^-_\rho(v)}(w')$, it is also removed from $\Gamma_{G^-_\rho}(w)$.  A vertex is removed from   $\Gamma_{G^-_\rho(v)}(w')$ if and only if all its members appear two or more times somewhere else in the graph.  By construction, for every $u$ in a clique in $\Gamma_{G^-_\rho(v)}(w')$, $A_u^v(s) \leq A_u(s)$ for every $s = 1, ..., \rho$.  This is because the vertices in $G^-_\rho(v)$ are a subset of the vertices in $G^-_\rho$.  Therefore, a vertex is removed from $\Gamma_{G^-_\rho(v)}(w')$, it is also removed from $\Gamma_{G^-_\rho}(w)$.

In the other direction, we prove that if a clique is removed from $\Gamma_{G^-_\rho}(w)$, it is also removed from $\Gamma_{G^-_\rho(v)}(w')$.  A clique is removed from $\Gamma_{G^-_\rho}(w)$ if and only if all its members appear twice elsewhere in the graph.  All cliques in $\Gamma_{G^-_\rho}(w)$ are composed by nodes at most $2\rho +1$ hops away from $v$.  Since every formed clique is at most diameter $\rho$, every appearance of nodes that compose the $\Gamma_{G^-_\rho}(w)$ can only occur in cliques with members that are at most $2\rho+1 +\rho = 3\rho +1$ from $v$.  Since $v$ has this amount of knowledge, if a clique is removed from $\Gamma_{G^-_\rho}(w)$, it is also removed from $\Gamma_{G^-_\rho(v)}(w')$.
\end{proof}


Since $\Gamma_{G^-_\rho}(w) = \Gamma_{G^-_\rho(v)}(w')$ and a node is removed from  $\Gamma_{G^-_\rho(v)}(w')$ if and only if it is also removed from $\Gamma_{G^-_\rho(v)}(w')$, for any arbitrary $w$ and $v$, then we have that the neighborhood of every $v\in nodes(w)$ for every $w \in G_\rho(v)$ is identical to the to the neighborhood of $w$ in $G_\rho$.
\end{proof}

In general, once the algorithm has been shown to be consistent we can describe the graph $G_\rho$ as the union of all $G_\rho(v)$ and each $G_\rho(v)$ is a local view subgraph of the $G_\rho$ graph centered around node $v$.  The algorithm is finalized by each node applying Kuhn's multicoloring to the graph $G_\rho(v)$.  The normalized sum-rate achieved by the Aggressive Selection Algorithm $\mathcal A_2(3\rho+1)$ with parameter $\rho$ is $\alpha_2(\rho) = \min_{v\in V} a(v)/\Delta_{G_\rho}$, where $a(v)$ is the number of vertices in $G_\rho(v)$ representing node $v$ and $\Delta_{G_\rho} = \max_{v \in V} \Delta_{G_\rho(v)}$.

\section{Step 3: Scheduling}
\label{sec:Step3}

In Step 2 of the algorithms, each user $v$ in the network finishes with a graph $G_\rho(v)$ that is composed of nodes representing $r$-cliques of at most diameter $\rho$ and at least one of those nodes includes user $v$.  The resulting $r$-cliques from Step 2 indicate that whenever user $v$ transmits, it will do so along with all the other users that are members of the $r$-cliques that include user $v$ using the optimal physical layer scheme.  This simultaneous transmission of users in an $r$-clique using optimal physical layer schemes is possible under the assumption of $\rho +1$ hops of channel information (since $r\leq \rho$) and is the opportunity for gain that our algorithms are leveraging.  In the third step of the proposed algorithms, users in the network determine the time-slots when their respective $r$-cliques have been assigned to transmit.  In other words, in Step 3, the graphs formed in Step 2 are scheduled.  One approach to schedule nodes in a graph is to use graph coloring \cite{diestel:2005}.  The problem of minimizing the number of required colors to color a graph, and thus increasing scheduling efficiency, has been a widely studied \cite{linial:1992,luby:1993, kothapalli:2006}.  


To perform scheduling in our algorithms, we use a local multicoloring algorithm introduced by Kuhn in \cite{kuhn:2009} since it is a one-shot algorithm and does not add to the number of hops of network information required for execution.  First, we describe the local multicoloring algorithm in terms of a normal graph and explain the performance that can be expected, then we go into the detail of how this algorithm is used in our Step~3 of our algorithms.

\subsection{Local Multicoloring Algorithm}
Consider a graph $G = (V, E)$ with $N$ nodes.  We assume each node knows the number of users, $N$, and a parameter, $k$.  The local multicoloring algorithm proceeds in three steps:

\begin{enumerate}
\item Each node $v\in V$ generates a vector $L_v = [l_{v,1}, l_{v, 2}, \ldots , l_{v, k}]$ of $k$ random numbers, where each $l_{v, i}$ is chosen uniformly from the set $\{1, 2, \ldots, kN^4\}$.
\item Each node $v$ sends the vector $L_v$ to all its neighbors.  We call the set of neighbors of node $v$, $\Gamma(v)$.  Each node $v$ also receives the vectors $L_u$, for all $u\in \Gamma(v)$
\item  Each node $v$ acquires all colors $i$ for which $l_{v, i} < l_{u,i}$, for all $u \in \Gamma(v)$.
\end{enumerate} 
The results in \cite{kuhn:2009} show that if $k$ is chosen to be greater that or equal to $6(\Delta +1) \ln(N)/\varepsilon^2$, then each node $v$ will acquire at least a fraction $\frac{1-\varepsilon}{\delta_v+1}$ of the $k$ colors available, with high probability.  It is important to note that one could relax the assumption of having to know $N$ and $k$ (which requires knowledge of $\Delta$) and only require knowledge of an upper bound on $N$, denoted as $\overline N$, and a predetermined, network-wide $\varepsilon$.  With this information, each node would know that the length of the random number vector it has to choose is $k = 6(\overline N +1) \ln(\overline N)/\varepsilon$ and each random number would be uniformly chosen from the set $\{1, 2, \ldots, k\overline N^4\}$.

\subsection{Application}
The local multicoloring algorithm can be used to schedule the sub-networks that have been formed in Step 2.  We begin by assuming each node $v$ in the original conflict graph $G$ generates a random number vector $L_v = [l_{v,1}, l_{v, 2}, \ldots , l_{v, k}]$ of $k$ random numbers, where each $l_{v, i}$ is chosen uniformly from the set $\{1, 2, \ldots, k\overline N^4\}$, where $k = 6(\overline N +1) \ln(\overline N)/\varepsilon^2$.

Consider the graph $G_\rho(v) = (W(v), F(v))$ from the point of view of node $v$.  We have assumed that each node has $\tau = 3\rho+3$ hops of connectivity information for the conservative algorithm and $\tau = 3\rho+1$ hops of connectivity information for the aggressive algorithm.  We now assume that in the process of information exchange to learn the connectivity of the network, the random vectors from all nodes $\tau$ hops away are also learned by each node.  This means that every user $v$ in the original conflict graph, $G$, knows all the random number vectors for all users in its $r$-clique(s) and all the random number vectors for all members of its neighboring $r$-cliques in $G_\rho(v)$. 

Each node $v$ finds the node with the smallest ID in each $w \in W(v)$.  Node $v$ assumes that the random number vector for vertex $w \in W(v)$ is the random number vector corresponding to the node with the smallest ID in that $r$-clique.  In other words, 
\begin{equation}
L_{w \in W(v)} = L_{\min\{x: x\in nodes(w)\}}.
\end{equation}

Once the random vectors for each vertex in $G_\rho(v)$ have been identified, we can apply the local multicoloring algorithm such that node $v$ knows the colors assigned to all $r$-cliques to which it belongs.  A vertex $w \in G_\rho(v)$ will be assigned time-slot $i$ if $l_{w, i} < l_{z,i}$, for all $z \in \Gamma(w)$.  If node $v$ is represented by vertex $w$, then node $v$ knows $L_w$ and $L_z$ for all $z \in \Gamma(w)$ because we assume each node knows the random number vectors of nodes $\tau$ hops away.  This also ensures that all $L_w$ are consistent over all $G_\rho(v)$.

Using the result from \cite{kuhn:2009}, it can be concluded that each $r$-clique represented by some vertex $w$ will be assigned a fraction at least $(1-\varepsilon)/(\delta_w +1)$ of the total $k$ time-slots assigned with high probability, where $\delta_w$ is the degree of vertex $w$ in the graph $G_\rho$.

\subsection{Overhead}
Let us analyze the overhead of the Step 3 in our algorithms in more detail.  First, we note that in terms of hops of information, Step~3 does not require any extra hops beyond the $\tau$ hops of connectivity we assumed in Step 2.  We have assumed that as the connectivity information is being exchanged the random number vectors required for local multicoloring are also being communicated.  The sharing of these vectors requires communication by each node of $6(\overline N +1) \ln(\overline N)/\varepsilon^2$ random numbers, each with a possible magnitude of up to $k\overline N^4$.  This results in generation and communication of $\mathcal O(\overline N \log^2(\overline N)/\epsilon^2)$ random bits by each node.  We can use the same non-trivial probabilistic argument mentioned in \cite{kuhn:2009} to claim the same results with only $\mathcal O(\log(\overline N))$ bits required.

Another important consideration to keep in mind is that we have assumed that the parameter $k$ can be chosen arbitrarily at the cost of complexity and amount of random bits to be exchanged.  Since $k$ represents the number of time-slots to be assigned assuming a static graph $G_\rho$, in practical applications, the value of $k$ might be constrained by the coherence time of the network.  In the description of our algorithms we have assumed that $k$ is smaller than the connectivity and channel coherence time of the network.

\section{Results}
\label{sec:Results}

We first characterize the generalized normalized sum-rate, $\myalpha(\rho)$, achieved by the previously presented conservative algorithm \cite{santacruz:2013} and the newly proposed aggressive algorithm with $(\eta,\tau)$ hops of network information.  Both algorithms conclude with a set of graphs, $G_\rho(v)$, $\forall v \in G$, and the performance of both proposed algorithms can be described in terms of the topology characteristics of the consolidated final graph, $G_\rho$.  The graph $G_\rho$ is the union over all graphs $G_\rho(v)$.  The normalized sum-rate performance of the algorithms is summarized in the following theorems.

\begin{theorem} \emph{(Generalized normalized sum-rate for conservative algorithm)} \cite{santacruz:2013}
Consider a conflict graph, $G$, where each user has $\eta$ hops of channel information and $\tau$ hops of connectivity information.  Using the conservative sub-network scheduling algorithm, the achievable normalized sum-rate is 

\begin{equation}
\myalpha_1(\eta, \tau) = \myalpha_1(\rho) = \frac{1-\varepsilon}{\Delta_{G_\rho} +1},
\end{equation}
with high probability, where $\Delta_{G_\rho}$ is the maximum degree of graph $G_\rho$ and $\varepsilon >0$.
\end{theorem}

\begin{proof}
We use the result from \cite{aggarwal} regarding the normalized sum-rate of independent graph scheduling.  The results says that if a network is divided into $t$ sub-graphs, $\mathcal A_1, ... \mathcal A_t$ (not all distinct, for some $t$) and each user $i$ belongs to $d_i$ independent sub-graphs, then the normalized sum-rate of the network is 
\begin{equation}
\min_{i \in 1, 2, ... N}\frac{d_i}{t}.
\end{equation}

In our sub-network scheduling algorithms, we have generated $k$ independent subgraphs.  A set of sub-networks that share one of the $k$ colors is an independent sub-graph since is composed by a set of sub-networks, each with a normalized sum-rate of 1, that do not interfere with each other.  

By properties of the local multicoloring algorithm, each sub-network $w \in G_\rho$ will be assigned $\left(\frac{1-\varepsilon}{\delta_w +1} \right)k$ colors in total.  Since in the conservative algorithm each user can only be represented by one sub-network, a user $j$  in sub-network $w$ appears in $d_j = \left(\frac{1-\varepsilon}{\delta_w +1} \right)k$ sub-graphs.  Therefore,

\begin{align}
\myalpha_1(\rho) &= \min_{i \in 1, 2, ... N} \frac{d_i}{t} \\
&= \min_{w \in G_\rho}\frac{\left(\frac{1-\varepsilon}{\delta_w +1} \right)k}{k} \\
&= \frac{1-\varepsilon}{\Delta_{G_\rho} +1}. 
\end{align}
\end{proof}

Similarly, we also describe the performance of the aggressive sub-network scheduling algorithm.

\begin{theorem}\emph{(Generalized normalized sum-rate for aggressive algorithm)}
Consider a conflict graph, $G$, where each user has $\eta$ hops of channel information and $\tau$ hops of connectivity information.  Using the aggressive sub-network scheduling algorithm, the achievable normalized sum-rate is 

\begin{equation}
\myalpha_2(\eta, \tau) = \myalpha_2(\rho) = \min_{v \in G} \sum_{v \in nodes(w)} \frac{1-\varepsilon}{\delta_w +1},
\end{equation}
with high probability.
\end{theorem}

\begin{proof}
In the case of the aggressive algorithm, each user can be represented in more than one sub-network.  Each sub-network $w$ will be active a total of $\left(\frac{1-\varepsilon}{\delta_w +1} \right)k$ time-slots.  Hence, the number of time-slots each user will be active is the sum of all the time-slots the sub-networks to which it belongs are active, in other words,

\begin{equation}
d_i = \sum_{i \in w} \left(\frac{1-\varepsilon}{\delta_w +1} \right)k.
\end{equation}
The worst-case node in terms of active time-slots gives us the normalized sum-rate of the network:

\begin{align}
\myalpha_2(\rho) &= \min_{i \in 1, 2, ... N} \frac{d_i}{t} \\
&= \min_{v \in G} \frac{\sum_{v \in w} \left(\frac{1-\varepsilon}{\delta_w +1} \right)k}{k} \\
&= \min_{v \in G} \sum_{v \in nodes(w)} \frac{1-\varepsilon}{\delta_w +1}. 
\end{align}
\end{proof}

Given the two evaluations of the generalized sum-rate of the proposed algorithms, we note that the main characteristic of the conservative algorithm is that the normalized sum-rate is ensured to be greater than or equal than the normalized sum-rate achieved by local multicoloring of the original network, $G$.  The guarantee provided by the conservative algorithm provides a proven improvement that leverages local information when $\rho \geq 1$.  This contribution is highlighted in the following theorem.

\begin{theorem}\emph{(Conservative algorithm guarantee)}
Let $\myalpha_1(\rho)$ be the normalized sum-rate of a network after applying Algorithm $\mathcal A_1(3\rho+3)$ to the original graph $G$.  If $\alpha(0)$ is the normalized sum-rate achieved by distributed multicoloring of the original network, $G$, then $\alpha(0)\leq \alpha_1(\rho)$, for $\rho \geq 1$.  The proof of this theorem is presented in \cite{santacruz:2013}.
\end{theorem}


We compare the conservative algorithm's performance to the distributed multicoloring algorithm of the original graph to highlight the advantages of leveraging local information.  The distributed multicoloring algorithm serves as a reasonable baseline of performance for one-shot algorithms.  In contrast, other algorithms such as distributed greedy scheduling \cite{wu:2007} or randomized maximal schedulers \cite{modiano:2006} consist of rounds of exchanges to make decisions.  By making our algorithms be one-shot algorithms, we ensure that the amount of knowledge required is constrained to a limited number of hops.  This quantifiable guarantee cannot be made under algorithms that involve several rounds as knowledge about the network propagates with each round.  We address the performance, overhead, and complexity of several algorithms more in detail in the next few subsections.

\subsection{Numerical Results: Normalized Sum-Rate}
We present numerical results that compare the performance of the Conservative and the Aggressive Selection Algorithms.  First, we present the performance of both algorithms in two example graphs for several values of the parameter $\rho = 0, 1, 2, 3$.  In these results, $\rho = 0$ reflects the case when there is no topology information and Kuhn's multicoloring algorithm is performed directly on the conflict graph, $G$.  The two sample graphs being compared are the $N$-node line-clique, which has been presented as an example throughout this paper, and the $N$-node line-star graph shown in Figure \ref{figs/linestar}.  
\begin{figure}[h]
  \centering
  \includegraphics[width=.2\textwidth]{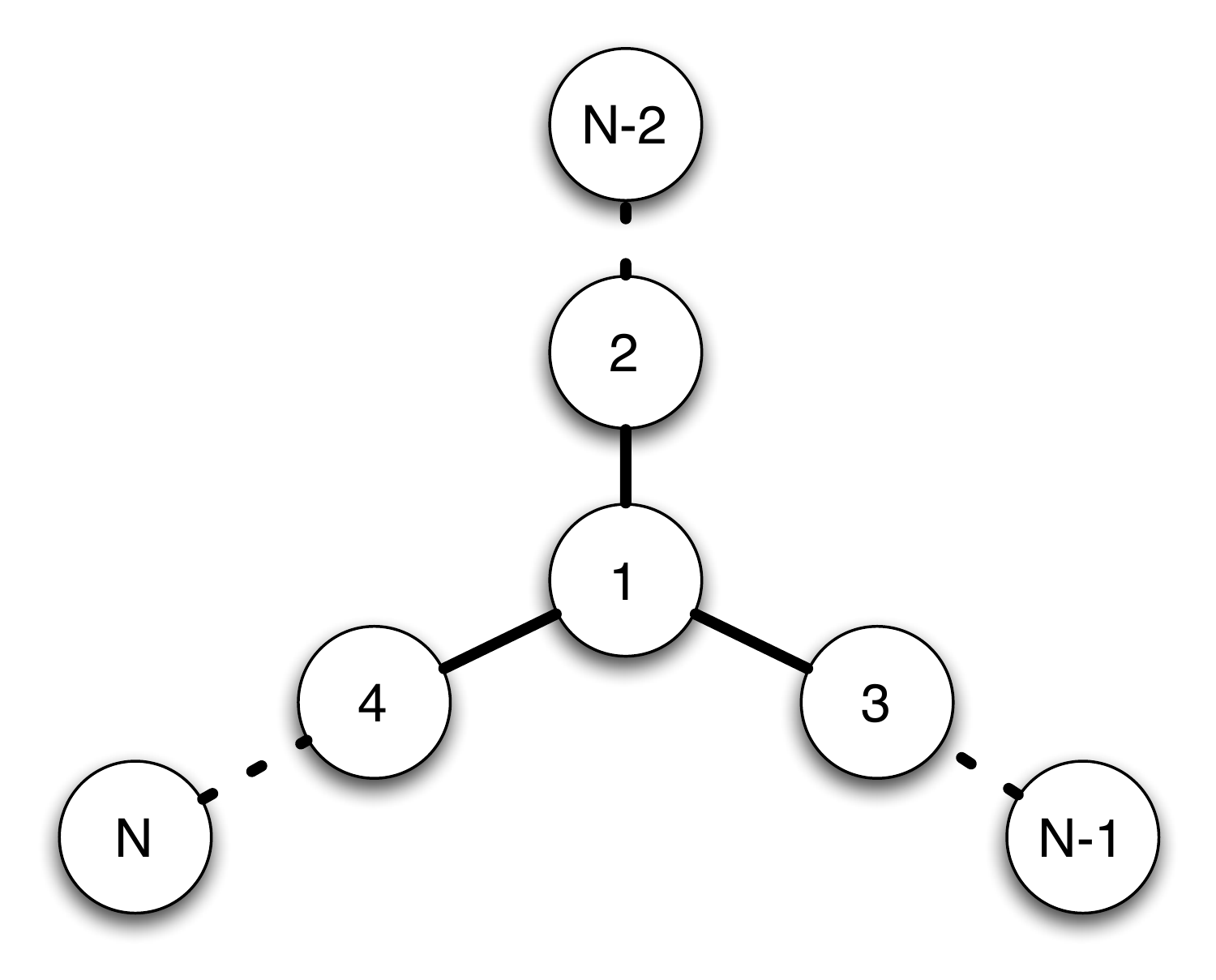}
  \caption{$N$-node Line-star Graph}
  \label{figs/linestar}
\end{figure}

The results shown in Figure \ref{figs/exgraphs} show that in both of these example graphs, the Aggressive Selection Algorithm outperforms the Conservative Selection Algorithm and the gain increases as the diameter of the cliques being formed increases.  The results on these sample graphs expose some of the limitations of the conservative algorithm, namely, the need for a \emph{unique} maximum $\rho$-clique to exist in order to form cliques.  In highly symmetrical graph such as the ones in these examples, the conservative algorithm provides marginal gains.  On the other hand, it is precisely in these situations where the aggressive algorithm displays its strengths.

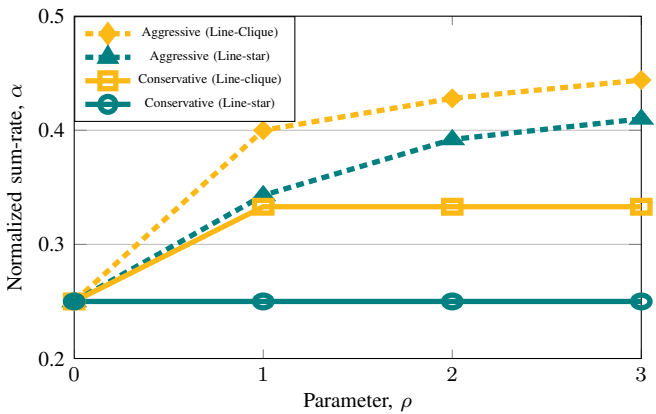
\begin{figure}
\centering
\begin{tikzpicture}[scale=1, font= \footnotesize]
	\begin{axis}[
		ymajorgrids,
		ylabel={Normalized sum-rate, $\alpha$},
		ylabel shift=-3pt,
		yscale=0.8,
		xscale=1.1,
		xtick={0,1,2,3},
		xlabel={Parameter, $\rho$},
		xmin = 0, xmax = 3, ymin = 0.2, ymax = 0.5,
		ytick={.2,.3,.4,.5},
		yticklabels={\footnotesize 0.2,\footnotesize 0.3,\footnotesize 0.4, \footnotesize 0.5},
		legend style = {font = \tiny,at={(0,1.25)}, anchor = north west},
		tick align = inside,
	]

	\addplot[yellow!70!red, mark = diamond*, mark size = 3pt, line width= 2pt, densely dashed, mark options = solid]
	coordinates {
	(0, 0.25)
	(1, 0.4)
	(2, 0.428)
	(3, 0.444)
	};

	\addplot[teal, mark = triangle*, mark size = 3pt, line width= 2pt, densely dashed, mark options = solid]
	coordinates {
	(0, 0.25)
	(1, 0.343)
	(2, 0.392)
	(3, 0.41)
	};
	
	\addplot[yellow!70!red, mark = square, mark size = 3pt, line width= 2pt]
	coordinates {
	(0, 0.25)
	(1, 0.333)
	(2, 0.333)
	(3, 0.333)
	};

	\addplot[teal, mark = o, mark size = 3pt, line width= 2pt]
	coordinates {
	(0,.25)
	(1,.25)
	(2,.25)
	(3,.25)
	};
	
	\legend{Aggressive (Line-Clique), Aggressive (Line-star), Conservative (Line-clique), Conservative (Line-star)}

	\end{axis}

\end{tikzpicture}

\caption{Conservative vs. Aggressive - Sample Graphs}
\label{figs/exgraphs}
\end{figure}


Given that our work builds on the difficulty of obtaining global information, we are especially concerned with small amounts of local information.  We are also interested in the algorithms' performance for classes of graphs that are representative of wireless network scenarios.  We present a comparison between the normalized sum-rate performance of four different algorithms for different classes of graphs and with parameter $\rho =1$ for the conservative and aggressive algorithms.  The algorithms selected for comparison are distributed coloring (DC), greedy scheduling (maximal scheduler, MS), our previous conservative algorithm (Con), and our aggressive algorithm (Agg).  The greedy distributed scheduling algorithm that produces maximal schedules is described as follows:
\begin{enumerate}
\item Assign a randomized ordering to the nodes in the network
\item Following the assigned order, a node is added to the schedule if it has packets to send and none of its interfering nodes have been scheduled
\end{enumerate}
Note that the greedy algorithm described here requires full network information.  There are distributed implementations of similar greedy scheduling algorithms that require rounds of communication with neighbors in the network.

The simulations are presented for three different classes of graphs.  We first look at Random Graphs, $\mathcal G(n, p)$, with $n$ nodes and edge probability $p$ in three regimes, $p = 0.1$ (low connectivity), $p=0.5$ (medium connectivity), and $p = 0.9$ (high connectivity).  Then we simulate algorithm performance in random scale-free graphs generated using the B-A algorithm~\cite{albert:2002}.  The degree distribution of scale-free graphs follows a power scale law and is a good representation of sensor networks.  The third class of graphs is random geometric graphs, in which $n$ transmitter-receiver pairs are randomly placed with uniform distribution in a unit square and interference occurs if any transmitter is within a diameter $d$ of a receiver from another user.

To evaluate algorithm performance, for each class of graphs and each parameter setting 100 independent random graphs are generated and the average normalized sum-rate of each algorithm is reported.  Figure \ref{figs/randomalphacomp} shows the performance comparison of random graphs with parameters $\mathcal G(20, 0.1)$, $\mathcal G(20, 0.5)$, and $\mathcal G(20, 0.9)$, scale-free graphs with 100 users, and geometric graphs with 20 users for parameters $d = 0.25$ and $d = 0.5$.  

\begin{figure}[ht]
\begin{tikzpicture}[scale = 1, font= \footnotesize]
	\begin{axis}[ylabel={Normalized sum-rate, $\alpha$ (avg)},
			ylabel shift=-3pt,
			ybar, 
			yscale=.9,
			bar width = 5pt,
			xtick={1,2,3,4,5,6},
			xticklabels={$p=0.1$, $p=0.5$,  $p=0.9$, Scale-free, $d=0.25$, $d=0.5$},
			xticklabel style={font=\tiny},
			xlabel={Graph type},
			xscale=1.1,
			ymin = 0, ymax = 0.4,
			ytick={0,.1,.2,.3,.4},
			legend style={at={(0.5,-0.25)}, anchor=north,legend columns=4}
]

\addplot+[teal,
		draw = black, 
  		error bars/.cd,
    		y dir=both,
    		y explicit,
    		error bar style={line width=.75pt},
    		error mark options={rotate=90, gray, mark size=3pt, line width=.75pt}
		]		
coordinates {
(1,0.1843) +- (1,0.0)
(2,0.07) +- (2, 0.0)
(3,0.0542) +- (3, 0.0)
(4,0.0585) +- (4,0.0)
(5,0.1221) +- (5, 0.0)
(6,0.0561) +- (6, 0.0)
};

\addplot+[yellow!70!red,
		draw = black,
  		error bars/.cd,
    		y dir=both,
    		y explicit,
    		error bar style={line width=.75pt},
    		error mark options={
      		rotate=90,
      		gray,
      		mark size=3pt,
      		line width=.75pt
    		}
		]		
coordinates {
(1, 0.217) +- (1, 0)
(2, 0.0954) +- (2, 0.0)
(3, 0.065) +- (3, 0.0)
(4, 0.0635) +- (4, 0)
(5, 0.1361) +- (5, 0.0)
(6, 0.0586) +- (6, 0.0)
};

\addplot+[blue!70!red,
		draw = black,
  		error bars/.cd,
    		y dir=both,
    		y explicit,
    		error bar style={line width=.75pt},
    		error mark options={
      		rotate=90,
      		gray,
      		mark size=3pt,
      		line width=.75pt
    		}
		]		
coordinates {
(1, 0.1846) +- (1, 0)
(2, 0.071) +- (2, 0.0)
(3, 0.081) +- (3, 0.0)
(4, 0.0591) +- (4, 0)
(5, 0.1289) +- (5, 0.0)
(6, 0.0943) +- (6, 0.0)
};

\addplot+[pattern color = red,
		draw = black,
		pattern= north east lines,
  		error bars/.cd,
    		y dir=both,
    		y explicit,
    		error bar style={line width=.75pt},
    		error mark options={
      		rotate=90,
      		gray,
      		mark size=3pt,
      		line width=.75pt
    		}
		]		
coordinates {
(1, 0.3841) +- (1, 0)
(2, 0.2957) +- (2, 0.0)
(3, 0.3154) +- (3, 0.0)
(4, 0.1307) +- (4, 0)
(5, 0.3205) +- (5, 0.0)
(6, 0.2909) +- (6, 0.0)
};
\legend{DC, MS, Con, Agg}

\end{axis}

\end{tikzpicture}

\caption{Algorithm performance comparison for random, scale-free, and geometric graphs}
\label{figs/randomalphacomp}
\end{figure}
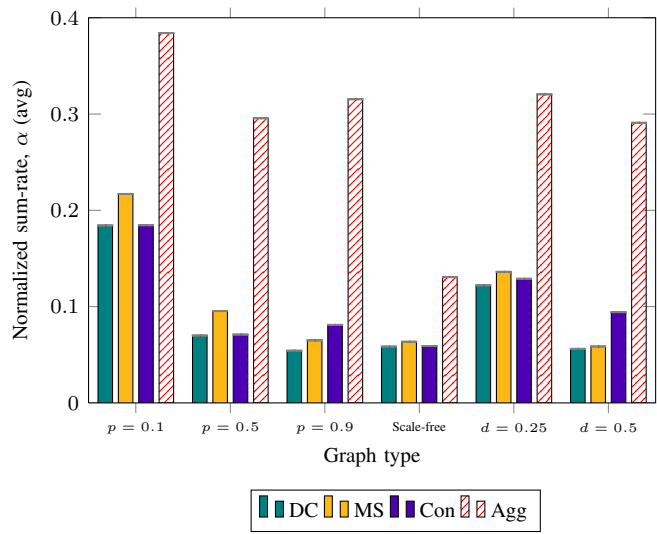

The aggressive algorithm outperforms the distributed coloring and greedy (maximal schedule) algorithms in all this cases.  Some notable points about the results include the performance of the conservative algorithm in the highly connected random graphs ($p = 0.9$) where cliques with large number of nodes are readily present and the conservative algorithm is able to outperform distributed coloring and maximal scheduling.  Also, note the performance in the case of scale-free graphs, where the conservative algorithm could often not ensure gains and so it remained conservative, and close in performance to distributed coloring, while the aggressive algorithm formed cliques for gains in normalized sum-rate.

%
%
%
%
%
%



\subsection{Net Sum-rate Comparisons}

We have introduced two distributed sub-network scheduling algorithms and analyzed their normalized sum-rate performance.  Performance in terms of normalized sum-rate is an important result because it represents the guaranteed fraction of what could be achieved with the optimal strategy and full knowledge.  Nevertheless, since it provides a guarantee, the normalized sum-rate characterizes worst-case performance behavior which could be significantly different from average performance in terms of other metrics of performance, for example \emph{net} sum-rate.  We simulate and analyze the performance of the conservative and aggressive algorithms and compare them to distributed coloring and maximal scheduling in terms of net sum-rate. 

Net sum-rate describes network performance without taking into account what the optimal, full-knowledge capacity of the network is.  While net sum-rate provides a practical measure of net throughput performance, characterization for arbitrary networks is unavailable due to the open problem of characterizing the capacity of the general interference channel.  The achievable sum rate within each one of the sub-networks scheduled after the application of the algorithms depends on the interference properties of the network.  In order to compute the achievable net sum-rate we must know the interference channel gains and the scheme being employed by each sub-network attempting to use advanced physical layer strategies.  For example, in the low interference regime, we can assume treating interference as noise while in the high interference regime we can assume interference cancellation.  In general, the achievable sum-rate for a set of interfering users is not known.  Nevertheless, within each one of the selected sub-networks, the achievable rate by each user can be lower and upper bounded.  Assume each user $i$ has a capacity $C_i$ when no interference is present.  In the worst case scenario, each member within each sub-network has to time-share the medium with the other users in the sub-network.  In this case, a user $i$ in a sub-network with $m$ users will achieve rate $C_i/m$ every time the sub-network is active, in average.  In the best case scenario, all interference is manageable and each user $i$ achieves $C_i$ every time the sub-network is active. 

Just as in the case of normalized sum-rate, we present net sum-rate performance comparison for simulation in three classes of graphs: random graphs, scale-free graphs, and geometric graphs.  In this sub-section, we report net sum-rate performance for four algorithms: distributed coloring (DC), greedy distributed algorithm (maximal schedule, MS), bounds of our conservative algorithm (Con), and bounds of our aggressive algorithm (Agg), both with $\rho = 1$.

Figure \ref{figs/rg_p01_net} shows the net sum-rate performance of random graph with parameters $\mathcal G(n, 0.1)$ and Figure \ref{figs/rg_p09_net} for graphs with $\mathcal G(n, 0.9)$.  In Figure \ref{figs/scalefreeresults_net}, we report the net sum-rate achievable for the class of scale-free graphs with 25, 50, and 100 users.  Finally, in Figure \ref{figs/geo_d025_net} we see the performance comparison in the class of geometric graphs with interference diameter $d = 0.25$.
	
\begin{figure}[ht]
\centering
\begin{tikzpicture}[scale=1]
	\begin{axis}[
		xscale=1.1,
		xtick={1,2,3},
		ymajorgrids,
		yscale=.8,
		ylabel={Net sum-rate},
		ylabel shift=-3pt,
		xticklabels={\footnotesize 5,\footnotesize 10,\footnotesize 20},
		xlabel={Number of users},
		ymin = 2.00, ymax = 12.00,
		ytick={2,4,6,8, 10, 12},
		legend style = {font = \footnotesize,at={(.025,1.25)}, anchor = north west},
		tick align = inside,
	]
		
	\addplot[teal, mark = diamond*, mark size = 3pt, line width= 1pt]
	coordinates {
	(1, 2.5450)
	(2, 6.6360)
	(3, 8.7922)
	};

	\addplot[yellow!70!red, mark = triangle*, mark size = 3pt, line width= 1pt]
	coordinates {
	(1, 2.6375)
	(2, 6.9216)
	(3, 11.027)
	};
	
	\addplot[blue!70!red, mark = square, mark size = 3pt, line width= 1pt]
	coordinates {
	(1, 2.4)
	(2, 5.3919)
	(3, 6.4599)
	};
	
	\addplot[blue!70!red, mark = square, mark size = 3pt, line width= 1pt, densely dashed, mark options = solid]
	coordinates {
	(1, 3.0733)
	(2, 6.4733)
	(3, 7.6010)
	};

	\addplot[red, mark = o, mark size = 3pt, line width= 1pt]
	coordinates {
	(1,2.5175)
	(2,6.2714)
	(3,8.5888)
	};
	
	\addplot[red, mark = o, mark size = 3pt, line width= 1pt, densely dashed, mark options = solid]
	coordinates {
	(1,2.8280)
	(2,7.2322)
	(3,11.1336)
	};
	
	\legend{DC, MS, Con (Lower), Con (Upper), Agg (Lower), Agg (Upper)}
	\end{axis}
\end{tikzpicture}

\caption{Net sum-rate performance comparison in random graphs with low connectivity, $p = 0.1$}
\label{figs/rg_p01_net}
\end{figure}
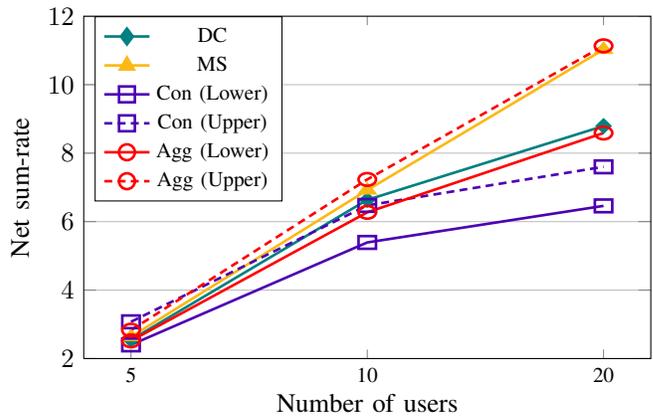
\begin{figure}[ht]
\centering
\begin{tikzpicture}[scale=1]
	\begin{axis}[
		xscale=1.1,
		xtick={1,2,3},
		ymajorgrids,
		yscale=0.8,
		ylabel={Net sum-rate},
		ylabel shift=-3pt,
		xticklabels={\footnotesize 5,\footnotesize 10,\footnotesize 20},
		xlabel={Number of users},
		ymin = .9, ymax = 10.00,
		ytick={2,4,6,8, 10},
		legend style = {font = \footnotesize,at={(.025,1.25)}, anchor = north west},
		tick align = inside,
	]
		
	\addplot[teal, mark = diamond*, mark size = 3pt, line width= 1pt]
	coordinates {
	(1, 1.0912)
	(2, 1.1176)
	(3, 1.3215)
	};

	\addplot[yellow!70!red, mark = triangle*, mark size = 3pt, line width= 1pt]
	coordinates {
	(1, 1.2875)
	(2, 1.6804)
	(3, 2.6730)
	};
	
	\addplot[blue!70!red, mark = square, mark size = 3pt, line width= 1pt]
	coordinates {
	(1, 0.9989)
	(2, 1.0533)
	(3, 1.2655)
	};
	
	\addplot[blue!70!red, mark = square, mark size = 3pt, line width= 1pt, densely dashed, mark options = solid]
	coordinates {
	(1, 3.5750)
	(2, 2.8100)
	(3, 1.8822)
	};

	\addplot[red, mark = o, mark size = 3pt, line width= 1pt]
	coordinates {
	(1, 1.4170)
	(2, 2.5112)
	(3, 4.9206)
	};
	
	\addplot[red, mark = o, mark size = 3pt, line width= 1pt, densely dashed, mark options = solid]
	coordinates {
	(1, 2.3006)
	(2, 4.4935)
	(3, 9.2197)
	};
	
	\legend{DC, MS, Con (Lower), Con (Upper), Agg (Lower), Agg (Upper)}
	\end{axis}
\end{tikzpicture}

\caption{Net sum-rate performance comparison in random graphs with high connectivity, $p = 0.9$}
\label{figs/rg_p09_net}
\end{figure}
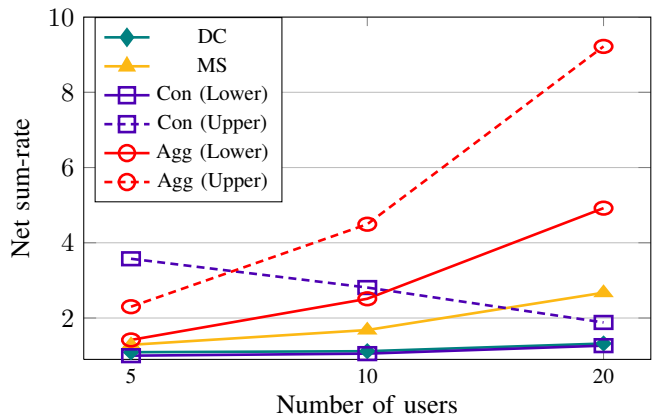
\begin{figure}[ht]
\centering
\begin{tikzpicture}[scale=1]
	\begin{axis}[
		xscale=1.1,
		xtick={1,2,3},
		ymajorgrids,
		yscale=0.8,
		ylabel={Net sum-rate},
		ylabel shift=-3pt,
		xticklabels={\footnotesize 25,\footnotesize 50,\footnotesize 100},
		xlabel={Number of users},
		ymin = 7, ymax = 65.00,
		ytick={10,20,30,40, 50, 60},
		legend style = {font = \footnotesize,at={(.025,1.25)}, anchor = north west},
		tick align = inside,
	]
		
	\addplot[teal, mark = diamond*, mark size = 3pt, line width= 1pt]
	coordinates {
	(1, 10.2218)
	(2, 20.7273)
	(3, 41.6845)
	};

	\addplot[yellow!70!red, mark = triangle*, mark size = 3pt, line width= 1pt]
	coordinates {
	(1, 15.0735)
	(2, 31.0506)
	(3, 62.9975)
	};
	
	\addplot[blue!70!red, mark = square, mark size = 3pt, line width= 1pt]
	coordinates {
	(1, 7.5939)
	(2, 14.5970)
	(3, 29.8690)
	};
	
	\addplot[blue!70!red, mark = square, mark size = 3pt, line width= 1pt, densely dashed, mark options = solid]
	coordinates {
	(1, 8.5491)
	(2, 16.5199)
	(3, 32.7765)
	};

	\addplot[red, mark = o, mark size = 3pt, line width= 1pt]
	coordinates {
	(1, 7.5939)
	(2, 17.7417)
	(3, 35.2784)
	};
	
	\addplot[red, mark = o, mark size = 3pt, line width= 1pt, densely dashed, mark options = solid]
	coordinates {
	(1, 11.8786)
	(2, 23.3875)
	(3, 46.2151)
	};
	
	\legend{DC, MS, Con (Lower), Con (Upper), Agg (Lower), Agg (Upper)}
	\end{axis}
\end{tikzpicture}

\caption{Net sum-rate performance comparison in scale-free graphs}
\label{figs/scalefreeresults_net}

\end{figure}
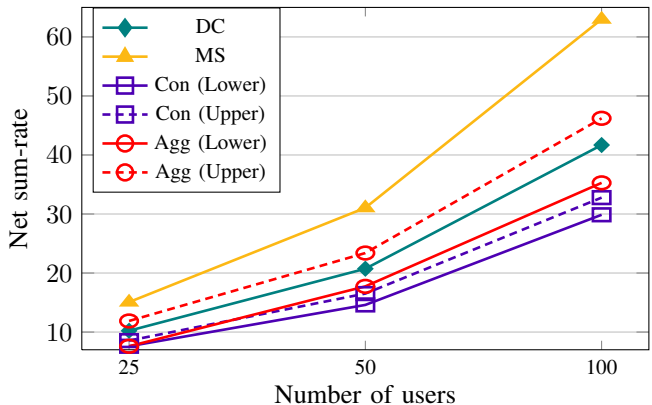
\begin{figure}[ht]
\centering
\begin{tikzpicture}[scale=1]
	\begin{axis}[
		xscale=1.1,
		xtick={1,2,3},
		ymajorgrids,
		yscale=0.8,
		ylabel={Net sum-rate},
		ylabel shift=-3pt,
		xticklabels={\footnotesize 10,\footnotesize 20,\footnotesize 30},
		xlabel={Number of users},
		ymin = 3, ymax = 14,
		ytick={3,6,9,12},
		legend style = {font = \footnotesize,at={(.025,1.25)}, anchor = north west},
		tick align = inside,
	]
		
	\addplot[teal, mark = diamond*, mark size = 3pt, line width= 1pt]
	coordinates {
	(1, 4.2115)
	(2, 4.9349)
	(3, 4.8420)
	};

	\addplot[yellow!70!red, mark = triangle*, mark size = 3pt, line width= 1pt]
	coordinates {
	(1, 4.8491)
	(2, 6.5589)
	(3, 7.1468)
	};
	
	\addplot[blue!70!red, mark = square, mark size = 3pt, line width= 1pt]
	coordinates {
	(1, 3.2104)
	(2, 3.4063)
	(3, 3.5068)
	};
	
	\addplot[blue!70!red, mark = square, mark size = 3pt, line width= 1pt, densely dashed, mark options = solid]
	coordinates {
	(1, 4.4835)
	(2, 4.7964)
	(3, 4.5393)
	};

	\addplot[red, mark = o, mark size = 3pt, line width= 1pt]
	coordinates {
	(1, 4.3019)
	(2, 6.0155)
	(3, 7.7804)
	};
	
	\addplot[red, mark = o, mark size = 3pt, line width= 1pt, densely dashed, mark options = solid]
	coordinates {
	(1, 5.7032)
	(2, 9.3040)
	(3, 13.1571)
	};
	
	\legend{DC, MS, Con (Lower), Con (Upper), Agg (Lower), Agg (Upper)}
	\end{axis}
\end{tikzpicture}

\caption{Net sum-rate performance comparison in geometric graphs, $d = 0.25$}
\vspace{-2mm}
\label{figs/geo_d025_net}
\end{figure}
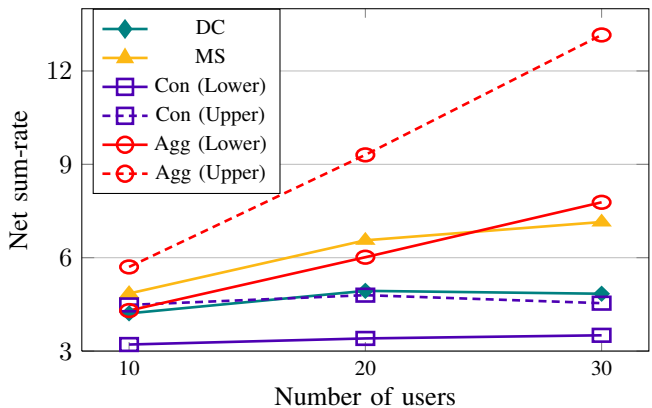
%
%
%
%
%

The results in Figures \ref{figs/rg_p01_net} - \ref{figs/geo_d025_net} show that in terms of net sum-rate, the aggressive algorithm outperforms the conservative algorithm in the majority of the cases.  This trend agrees with the performance in terms of normalized sum-rate.  An interesting observation from the numerical results is that the upper bound on the sum-rate achievable by our aggressive algorithm is consistently higher than the sum-rate achieved by distributed scheduling and the maximal scheduler for the simulated random graphs and geometric graphs.  This validates the assertion that when interference is manageable, our aggressive algorithm can perform better in terms of net sum-rate.  The only case where neither of our algorithms could surpass maximal scheduling was the class of scale-free graphs.  The structure of scale-free networks is particularly unfavorable to the aggressive algorithm since the spoke-hub nature of the graphs results in large degree increase when $\rho = 1$.  Nevertheless, in scale-free graphs, increasing the diameter of the cliques being formed to $\rho = 2$ can produce significant improvement since large sections of each spoke-hub are identified as a single sub-network.  The numerical results for net sum-rate performance highlight the advantages of the aggressive algorithm because, with only $\rho = 1$, it is able to outperform maximal scheduling in random and geometric graphs.  In cases where higher net sum-rate is not achieved, such as in scale-free graphs, it has the flexibility to leverage more knowledge (e.g., let $\rho = 2$) and improve performance.

\subsection{Remarks on Overhead}

The key feature of the algorithms we have proposed in our previous work and in this paper is that they can be executed with $\rho+1$ hops of channel information and $3\rho+3$ hops of connectivity information for the conservative algorithm or $3\rho +1$ hops of connectivity information for the aggressive algorithm.  In order to be able to make comparisons with other distributed scheduling algorithms, it is important to understand what is the overhead to required obtain this amount of information.  

The total amount of overhead required for execution depends on how often the system needs to renew its knowledge.  Consider three different time-scales in the life of the network.  First, we define \emph{topology coherence time}, $T_{topo}$, as the amount of time the network remains static in terms of topology.  In other words, the amount of time the network is correctly described by the graph $G$.  Second, the amount of time the channel remains constant before changing is denoted \emph{channel coherence time}, $T_{channel}$.  Finally, the amount of time required for a single communication transmission is denoted \emph{data transmission time}, $T_{data}$.  The length of each one of these time-scales and their ratios depend on network properties such as mobility, fading environments, and packet lengths.  Due to the nature of wireless networks and their physical properties, it is reasonable to assume that in most cases $T_{topo} \geq T_{channel} \geq T_{data}$.

When we describe our proposed algorithms as one-shot algorithms, we are referring to the fact that they need to exchange information only after some time $T_{topo}$ or $T_{channel}$.  This important feature contrasts with a large section of distributed link scheduling algorithms which require several rounds of message passing or dedicated control sub-frames each time a data communication is established, i.e., every $T_{data}$.  The distributed link scheduling algorithms described in Section \ref{sec:Related} require rounds of message passing per data transmission time because they are based on queue state information.  The queue state information of each user is updated after each data transmission and must be communicated every $T_{data}$.  While some of these algorithms only require communication from each node with neighbors only one hop away, they often require a large amount of rounds per $T_{data}$.  A sequence of message exchanges implicitly propagates information beyond one hop.

Direct overhead comparisons between the proposed algorithms and state-of-the-art distributed scheduling algorithms are non-trivial.  The primary objective of the proposed algorithms is a reduction in the number of hops of network information required for algorithm execution.  In other words, we are concerned with information locality.  On the other hand, distributed scheduling algorithms have their main priority set on reducing computational complexity \cite{modiano:2006, sanghavi:2007, sharma:2006, wu:2007, lin:2006, gupta:2009}.  Because the emphasis of efficiency is different in both scenarios, the relative cost of overhead also changes.  Therefore, while several rounds of message passing with neighbors is typical in distributed scheduling algorithms and considered low overhead, in terms of information locality, it can represent significant overhead since each round of message passing can implicitly provide an extra hop of network information.

\section{Conclusion}
\label{sec:Conclusion}

The work presented is motivated by the original question of how and when users in a wireless network should transmit if only local information is available.  To take steps towards solving this challenging problem, we have developed a distributed algorithm that use only local connectivity information to coordinate sub-networks that have enough channel information to communicate in an optimal manner.  The key idea behind the algorithm presented is to identify independent sub-graphs such that, at each time slot, the active users in the network have all the information required to use physical layers beyond interference avoidance.  

The previously proposed conservative algorithm is guaranteed to have a normalized sum-rate performance that is greater than or equal to distributed graph coloring.  The aggressive algorithm introduced in this work shows significant improvement for important graph classes over distributed graph coloring, maximal scheduling, and the conservative algorithm.  These two algorithms are constructive and provide more feasible schemes that can leverage realistic assumptions of local information.  Furthermore, we have shown the performance of our algorithms in terms of net sum-rate as well.  Although there are topologies for which maximal scheduling might provide higher net sum-rate, in general, the interference-avoidance approach is not always optimal and our aggressive algorithm can produce significant gains in topologies where local information can be leveraged, such as random graphs and geometric graphs.  An implication of the results presented in this paper is that the algorithms proposed provide a lower bound on the normalized sum-capacity with local information.  



%

%
%
%
%
%
%

\ifCLASSOPTIONcaptionsoff
  \newpage
\fi



%
\IEEEpeerreviewmaketitle

\bibliographystyle{IEEEtran}
\bibliography{IEEEabrv,references}

%








\end{document}